\documentclass[11pt]{article}
\usepackage{amssymb, amsmath, amsthm}
\usepackage{natbib}
\usepackage[colorlinks]{hyperref}
\usepackage{comment,url,algorithm,algorithmic,graphicx,subcaption,relsize}
\usepackage{amsfonts,mathtools,mathrsfs,mathtools,color}
\usepackage{xcolor}
\usepackage{array}

\makeatletter
\newtheorem*{rep@theorem}{\rep@title}
\newcommand{\newreptheorem}[2]{%
\newenvironment{rep#1}[1]{%
 \def\rep@title{#2 \ref{##1}}%
 \begin{rep@theorem}}%
 {\end{rep@theorem}}}
\makeatother

\newtheorem{definition}{Definition}
\newtheorem{theorem}{Theorem}
\newtheorem*{theorem*}{Theorem}
\newtheorem{lemma}{Lemma}

\newtheorem{remark}{Remark}
\newreptheorem{theorem}{Theorem}
\newreptheorem{lemma}{Lemma}

\newcommand{\B}{\mathcal{B}}
\newcommand{\E}{\mathbb{E}}
\renewcommand{\H}{\mathcal{H}}
\newcommand{\I}{\mathcal{I}}

\newcommand{\N}{\mathcal{N}}
\newcommand{\bN}{\mathbb{N}}
\renewcommand{\P}{\mathcal{P}}
\newcommand{\R}{\mathbb{R}} 
\newcommand{\sR}{\mathsf{R}}

\newcommand{\V}{\mathcal{V}}
\newcommand{\X}{\mathcal{X}}
\newcommand{\Y}{\mathcal{Y}}
\newcommand{\Z}{\mathcal{Z}}

\newcommand{\Cov}{\mathsf{Cov}}
\newcommand{\Var}{\mathsf{Var}}
\newcommand{\KL}{\mathsf{KL}}
\newcommand{\TV}{\mathsf{TV}}
\newcommand{\op}{\mathsf{op}}
\newcommand{\one}{\mathbf{1}}
\newcommand{\Vol}{\mathsf{Vol}}
\newcommand{\dist}{\mathsf{dist}}
\newcommand{\area}{\mathsf{area}}
\newcommand{\CPI}{C_{\textsf{PI}}}
\newcommand{\Succ}{\mathsf{Succ}}

\newcommand{\CCh}{C_{\textsf{Ch}}}

\renewcommand{\d}{\mathrm{d}}
\newcommand{\du}{\mathrm{d}u}
\newcommand{\dr}{\mathrm{d}r}
\newcommand{\ds}{\mathrm{d}s}
\newcommand{\dt}{\mathrm{d}t}
\newcommand{\dx}{\mathrm{d}x}
\newcommand{\dy}{\mathrm{d}y}

\title{The Geometry of Efficient Nonconvex Sampling}

\author{Santosh S. Vempala\thanks{Georgia Institute of Technology, College of Computing. Email: \texttt{vempala@gatech.edu}. This work was supported in part by NSF award CCF-2504994 and a Simons Investigator award.} 
\and
Andre Wibisono\thanks{Yale University, Department of Computer Science. Email: \texttt{andre.wibisono@yale.edu}. This work was supported by NSF awards CCF-2403391 and CAREER CCF-2443097.}}

\begin{document}

\maketitle

\vspace{1in}

\begin{abstract}%
    We present an efficient algorithm for uniformly sampling from an arbitrary compact body $\X \subset \R^n$ from a warm start under isoperimetry and a natural volume growth condition.
    Our result provides a substantial common generalization of known results for convex bodies and star-shaped bodies. 
    The complexity of the algorithm is polynomial in the dimension, the Poincar\'e constant of the uniform distribution on $\X$ and the volume growth constant of the set $\X$. %
\end{abstract}

\newpage
\setcounter{tocdepth}{2}
\tableofcontents

\section{Introduction}
\label{Sec:Intro}

Sampling from a bounded set $\X$ in a high-dimensional space $\R^n$ is a classical problem with connections to many topics in mathematics and theoretical computer science. 
The literature for the problem is largely based on the theory of sampling from a \textit{convex} body $\X$ (see Figure~\ref{Fig:Example}(a) for an example).
The celebrated work of Dyer, Frieze and Kannan~\citep{DFK91random} showed that convex bodies $\X$ can be sampled efficiently, i.e., in time polynomial in the dimension $n$, with only membership oracle access to the body $\X$ and a starting point inside the body. 
\citet{AK91sampling} extended  the frontier of polynomial-time sampling to the class of logconcave probability distributions, which can be viewed as the functional generalization of convex sets.
Subsequent works~\citep{LS90mixing,DyerF90,LS93,KLS97,LV07geometry,LV06hit,kookV2024} improved the complexity by introducing new ideas and refining the analyses.
The state-of-the-art result by Kook, Vempala and Zhang~\citep{kookV2024,kook2026and} provides an algorithm (\textbf{In-and-Out}) with an $\tilde O(n^2)$ iteration complexity for sampling from a (near-)isotropic \textit{convex} body $\X \subset \R^n$ with a warm start, where each iteration uses one call to the membership query to $\X$ and $O(n)$ arithmetic operations, with a guarantee in R\'enyi divergence. 
See also related work of~\citep{jia2026} for isotropic transformation, ~\citep{kookV2025} for sampling from general logconcave distributions, and~\citep{kook2025faster} for improved complexity of generating a warm start.

However, few results are known about the algorithmic complexity of sampling from \textit{nonconvex} bodies.
The work of Chandrasekaran, Dadush and Vempala~\citep{CDV10} showed that \textit{star-shaped} bodies $\X$ (see Figure~\ref{Fig:Example}(b) for an example) can be sampled with iteration complexity polynomial in the dimension $n$, in the inverse fraction of volume taken up by the convex core (the non-empty subset of the star-shaped body that can ``see'' all points in the body), and in $\varepsilon^{-1}$ where $\varepsilon > 0$ is the final error parameter in total variation distance specified in the input.
Beyond this result, we are not aware of any rigorous results on the algorithmic complexity of uniformly sampling nonconvex bodies. 
Indeed, sampling from nonconvex sets is intractable in the worst case~\citep{koutis2003hardness}.
Nevertheless, there are many ``reasonable'' nonconvex sets---such as those depicted in Figures~\ref{Fig:Example}(c) and~\ref{Fig:Example}(d)---that one should be able to sample in polynomial time, but such a statement does not follow from the existing theory, since the sets are not star-shaped.

\begin{figure}[t!]
    \centering
    \subfloat[Convex]{\includegraphics[width=.22\textwidth]{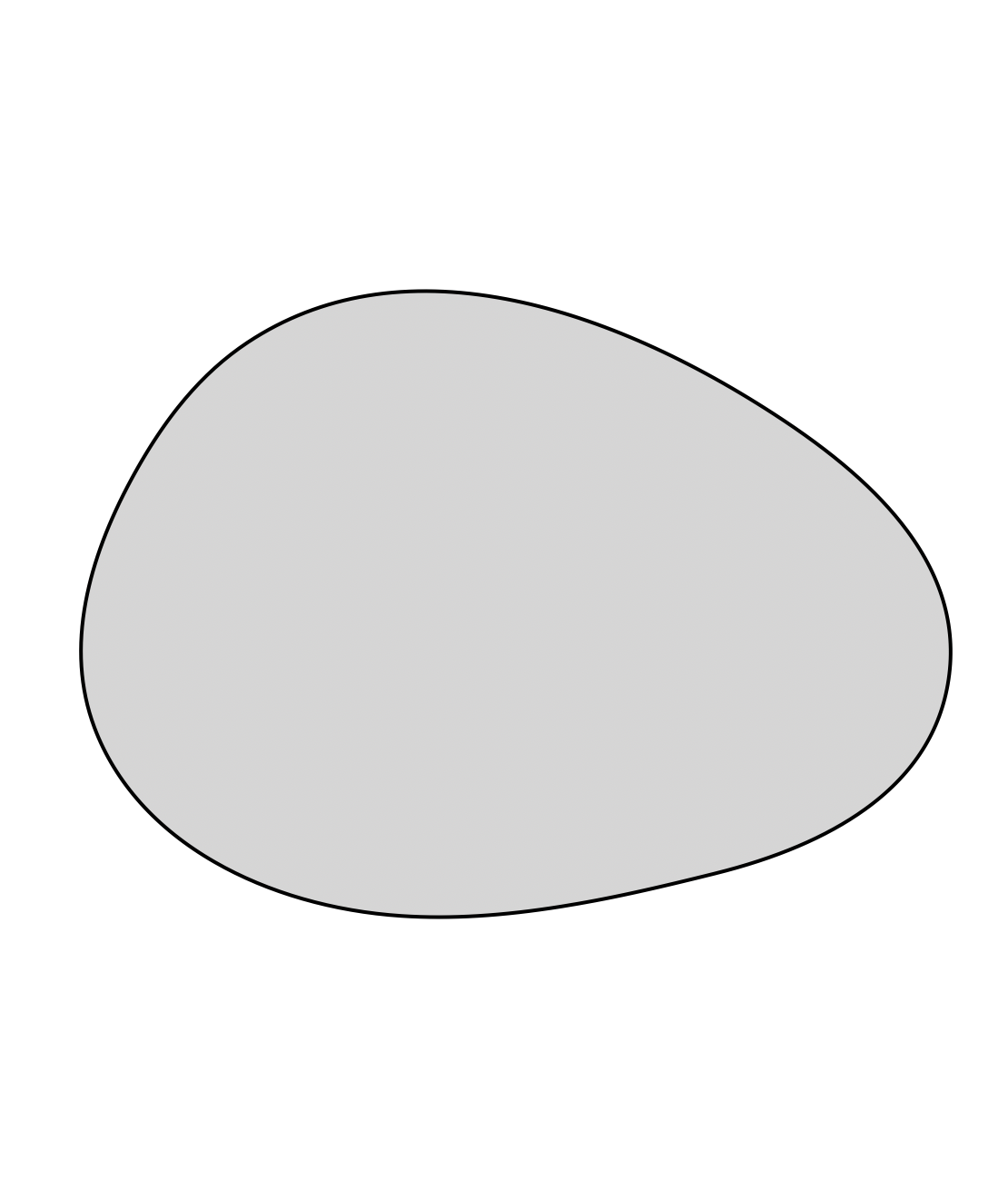}}
    \hfill
    \subfloat[Star-shaped]{\includegraphics[width=.22\textwidth]{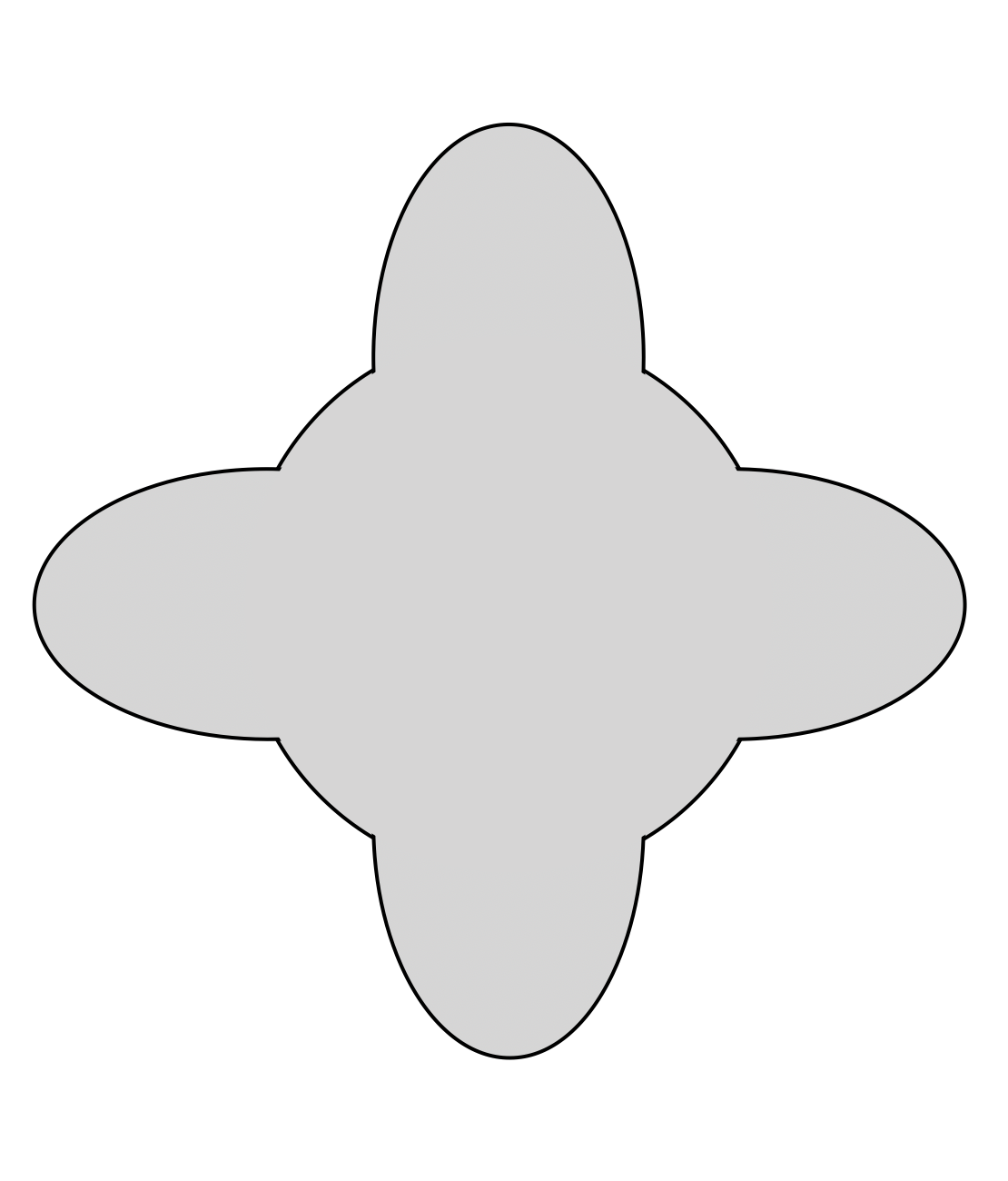}}
    \hfill
    \subfloat[Nonconvex with a hole]{\includegraphics[width=.22\textwidth]{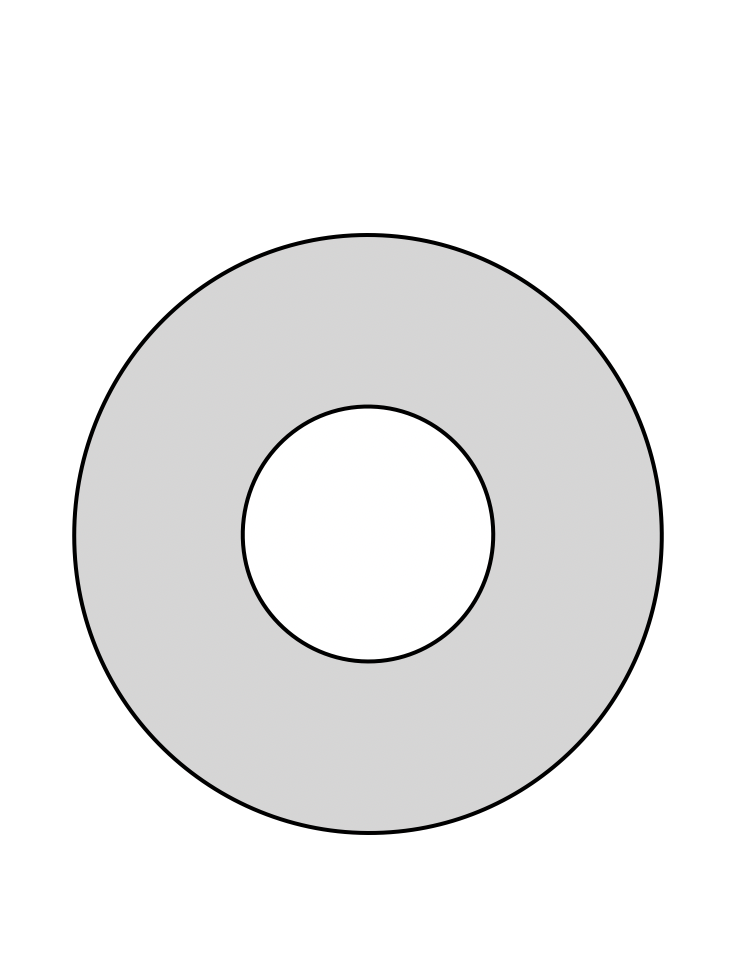}}
    \hfill    
    \subfloat[Non-star-shaped]{\includegraphics[width=.22\textwidth]{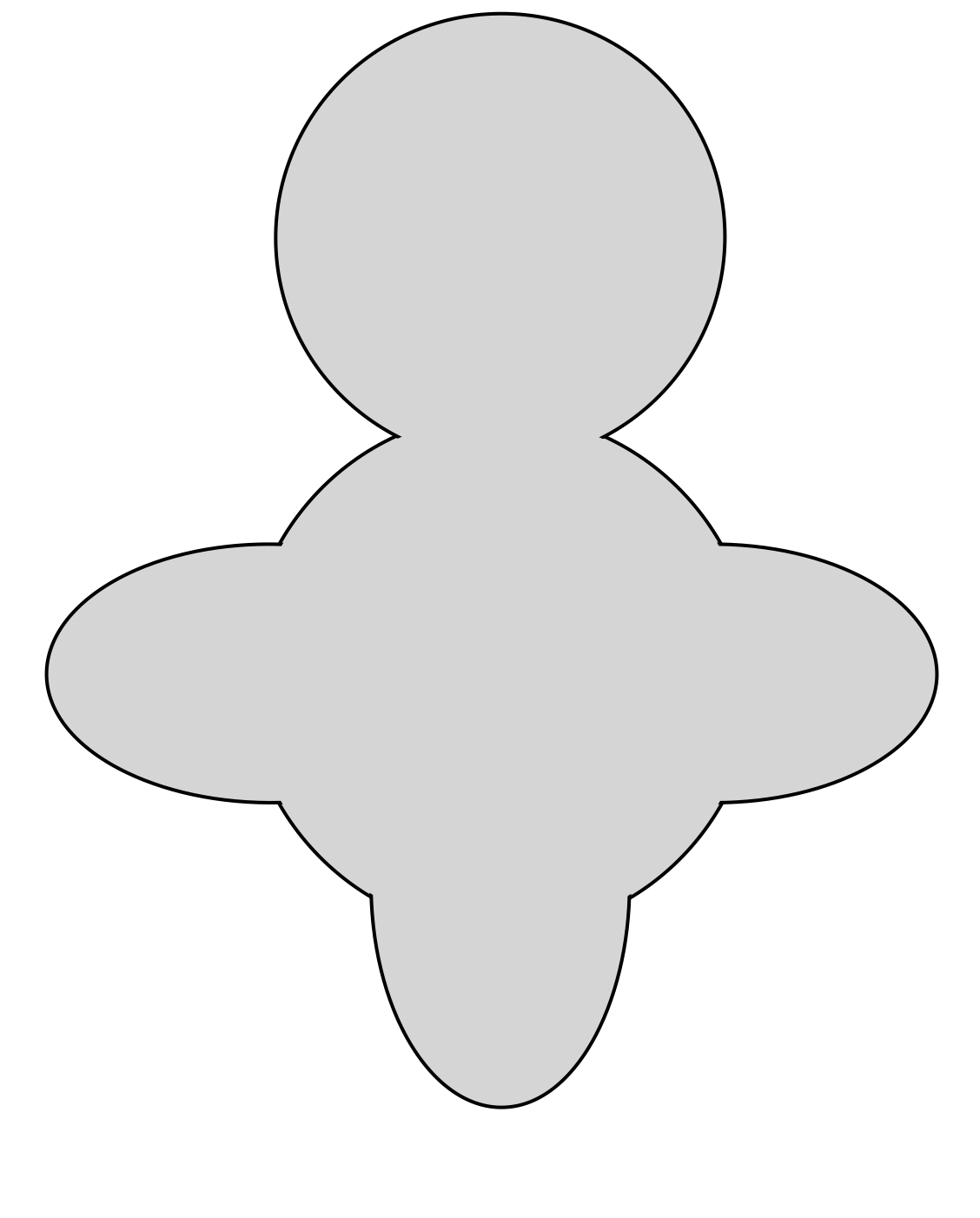}}
    \hfill
    \caption{Examples of bodies $\X \subset \R^n$ that current theory covers ((a) convex and (b) star-shaped); and examples of $\X$ that current theory does not cover ((c) with a hole and (d) not star-shaped).}
    \label{Fig:Example}
\end{figure}

In this paper, we provide an efficient algorithm for uniformly sampling from an arbitrary compact set $\X \subset \R^n$ under two minimal assumptions: (1) The uniform distribution on $\X$ satisfies isoperimetry, namely a Poincar\'e inequality; and 
(2) the set $\X$ satisfies a natural \textit{volume growth condition}, which we define below. 
These conditions substantially generalize convexity and star-shapedness, the two broad families previously known to be efficiently sampleable (and cover the examples in Figures~\ref{Fig:Example}(c) and~\ref{Fig:Example}(d), and hence our results show that we can indeed sample them in polynomial time). 
Here by \textit{efficient/polynomial-time}, we mean polynomial in the dimension $n$, the Poincar\'e constant of the target distribution, the volume growth constants of $\X$, and in $\log \varepsilon^{-1}$ where $\varepsilon > 0$ is the final error parameter in R\'enyi divergence specified in the input.
Therefore, our result recovers the classical polynomial-time sampling from convex bodies, and it
improves the result of~\cite{CDV10} to a polynomial-time sampling from star-shaped bodies (improving the dependence from $\mathrm{poly}(\varepsilon^{-1})$ to $\mathrm{poly}(\log \varepsilon^{-1})$) with stronger error guarantees in R\'enyi divergence.
We discuss these assumptions and the algorithm in more detail.

\paragraph{Isoperimetry.}
The analysis of sampling for convex bodies and logconcave densities reveals that the \textit{isoperimetry} of the target distribution is a crucial ingredient for efficient sampling.
For example, if a distribution has poor Cheeger isoperimetry (i.e., large Poincar\'e constant), then the domain can be partitioned into two subsets with a small separating boundary (see Figure~\ref{Fig:Example-2}(a) for an example), and any ``local'' process such as one based on diffusion, without global knowledge, would need many steps to even cross from such a subset to its complement. 

Our main question is thus the following:
\begin{center} 
{\em Does isoperimetry (Poincar\'e inequality) suffice for polynomial-time sampling?}
\end{center}

We recall that from the perspective of sampling via continuous-time diffusion processes such as the Langevin dynamics, isoperimetry (Poincar\'e inequality) is a natural sufficient condition for efficient sampling; see for example~\cite[Theorem~3]{VW23rapid}
for the convergence rate of the continuous-time Langevin dynamics under Poincar\'e inequality.
In discrete time, the \textit{Proximal Sampler} algorithm~\citep{LST21structured} 
has a convergence guarantee under Poincar\'e inequality analogous to the continuous-time Langevin dynamics; see~\cite[Theorem~4]{CCSW22improved} for the case of the unconstrained target distribution, and~\cite[Theorem~2]{kook2026and} for the case of the uniform distribution on $\X$ (see also Lemma~\ref{Lem:PSOuterConvergence}).
However, the \textit{Proximal Sampler} is an idealized algorithm, since it requires sampling from a regularized distribution in each iteration (see Section~\ref{Sec:ProxSampler} for a review).
The work of~\cite{CCSW22improved} shows how to implement the \textit{Proximal Sampler} as a concrete algorithm using rejection sampling when the target distribution $\pi \propto \exp(-f)$ has full support on $\R^n$ and $f$ is smooth (has a bounded second derivative).
Subsequent works~\citep{liang2022proximal,liang2023proximal,FYC23improved} show how to obtain an improved complexity using approximate rejection sampling under weaker smoothness assumption such as H\"older continuity of $\nabla f$, or Lipschitzness of $f$.
However, for our setting of the uniform target distribution $\pi \propto \one_\X$, the potential function $f$ is not even continuous: $f(x) = 0$ if $x \in \X$, and $f(x) = +\infty$ if $x \notin \X$; therefore, none of the results above apply.
When $\pi \propto \one_\X$ and $\X$ is \textit{convex}, the work of~\cite{kook2026and} shows how to implement the \textit{Proximal Sampler} via rejection sampling with a threshold on the number of trials, resulting in the \textbf{In-and-Out} algorithm that they analyze and prove has  $\tilde O(n^2)$ complexity from a warm start in an isotropic convex body.
The work of~\cite{dang2025oracle} studies a variant of \textbf{In-and-Out} using either the projection oracle or the separation oracle to the convex body.
However, when $\X$ is not convex (or star-shaped with a large convex core), there is currently no such algorithmic guarantee for sampling from $\X$; we address this gap in this paper.

\begin{figure}[t!]
    \centering
    \subfloat[A dumbbell-shaped body.]{\includegraphics[width=.35\textwidth]{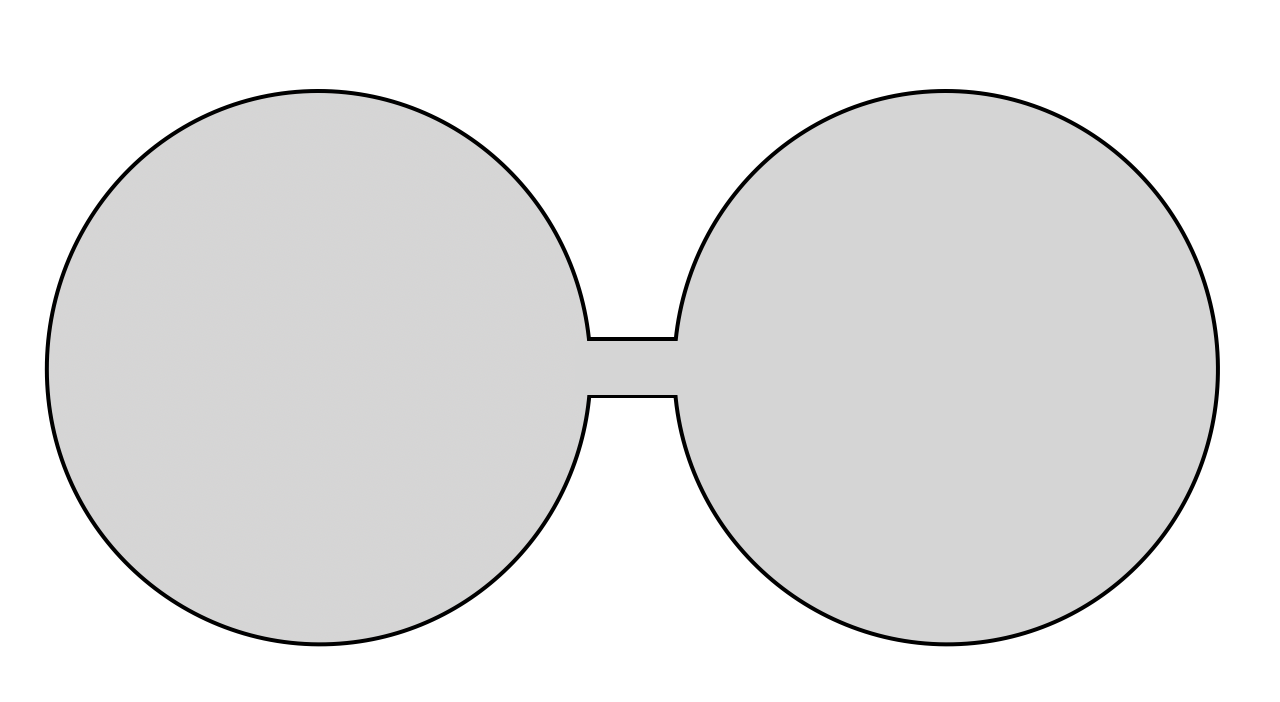}}
    \qquad\qquad
    \subfloat[A cylinder.]{\includegraphics[width=.4\textwidth]{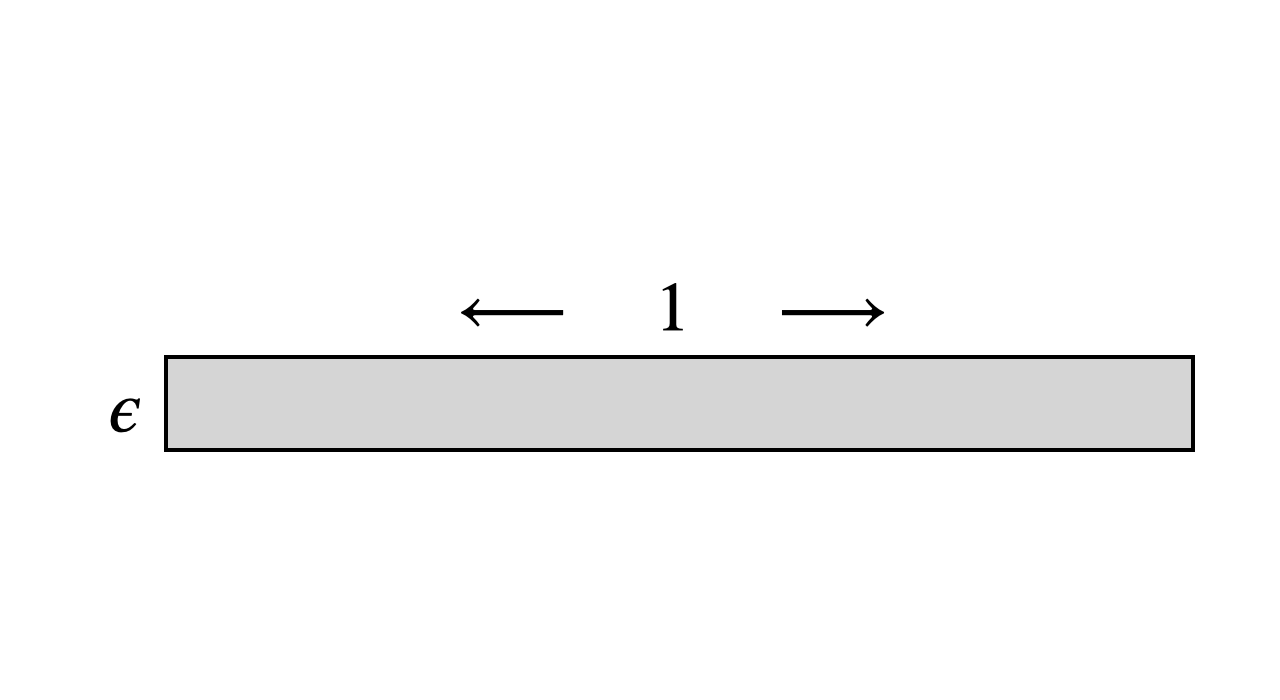}}
    \caption{A body with (a) poor isoperimetry (large Poincar\'e constant), and (b) good isoperimetry (small Poincar\'e constant) but large volume growth rate.}
    \label{Fig:Example-2}
\end{figure}

\paragraph{Volume growth condition.}
For sampling from $\X \subset \R^n$, it turns out that isoperimetry (Poincar\'e inequality) alone is not sufficient, and we need another condition. 
To illustrate, consider a cylinder in $\R^n$ of fixed axis length (say $1$), and base radius $\epsilon > 0$ (see Figure~\ref{Fig:Example-2}(b)). 
The uniform distribution on this cylinder has good isoperimetry, i.e., the Poincar\'e constant is $O(1)$. 
However, as the radius shrinks ($\epsilon \to 0$), any known Markov chain would need more membership queries to go from one end of the cylinder to another. 
Intuitively, any process that only has a local picture of the body to be sampled either makes very small steps or has a large probability of stepping out. 
We can capture this consideration via the \textit{volume growth rate} of $\X$ (see Section~\ref{Sec:VolGrowth} for a precise definition), which is how fast the volume of $\X \oplus B_t$ grows as a function of $t$, where $B_t \equiv B(0,t)$ is the $\ell_2$-ball of radius $t > 0$, and $\oplus$ is the Minkowski sum between sets.
When $\X$ is convex, the volume growth rate can be controlled by the \textit{external isoperimetry} of $\X$, which is the ratio of the surface area to the volume of $\X$ (see Lemma~\ref{Lem:VolGrowthConvex}); the external isoperimetry can in turn be controlled by the reciprocal of the radius of the largest $\ell_2$-ball contained in $\X$ (the ``inner radius'' of $\X$).
In the cylinder example above (Figure~\ref{Fig:Example-2}(b)), which is convex, the volume growth rate of $\X$ scales as $1/\epsilon$, which tends to $+\infty$ as $\epsilon \to 0$, suggesting that sampling from $\X$ in this case is difficult. 

This leads us to a more precise question:

\begin{center}
{\em Can the complexity of sampling from $\X \subset \R^n$ be bounded by a polynomial in the dimension $n$, the Poinc\'are constant of $\pi \propto \one_\X$ and the volume growth rate of $\X$?}
\end{center}

In this paper, we answer this question affirmatively.
We analyze the same \textbf{In-and-Out} algorithm studied by~\cite{kook2026and}, which only requires a membership oracle to $\X$.
Whereas~\cite{kook2026and} proved their result assuming $\X$ is convex, we prove our result \textit{without} assuming convexity, only isoperimetry and volume growth as stated above.
We know that isoperimetry holds for a large class of nonconvex sets.
Similarly, we show that the volume growth condition is preserved under operations including union and set subtraction, and so it captures a large class of nonconvex sets.
Therefore, our main result shows that a large class of nonconvex sets can be sampled in polynomial time, substantially generalizing previous results for convex and star-shaped sets.

\subsection{Algorithm: In-and-Out}

The \textbf{In-and-Out} algorithm~\citep{kook2026and} is the following iteration.

\begin{algorithm}[H]
\hspace*{\algorithmicindent} 
\begin{algorithmic}[1] 
\caption{: \textbf{In-and-Out} for sampling from $\pi \propto \one_\X$. \\
\textbf{Input:} Initial point $x_0 \in \X$, 
step size $h > 0$, number of steps $T \in \mathbb{N}$, threshold $N \in \mathbb{N}$. \\
\textbf{Output:} $x_T \sim \rho_{T}$.
} 
\label{Alg:In-and-Out}

\FOR{$i=0,\dotsc,T-1$}

\STATE Sample $y_{i}\sim
\N(x_{i},hI_{n})$.
\label{Eq:AlgFwd}

\STATE \textbf{Repeat:} Sample $x_{i+1} \sim \N(y_{i}, h I_{n})$
until $x_{i+1}\in \X$. 
If $\#$attempts$_{i}\, \ge N$, declare
\textbf{Failure}.\label{Eq:AlgBwd}

\ENDFOR
\end{algorithmic}
\end{algorithm}

The algorithm has three parameters: 
(1) the step size $h > 0$,
(2) the number of steps (or outer iterations) $T \in \mathbb{N}$, and
(3) the threshold on the number of trials $N \in \mathbb{N}$ in the rejection sampling (Step 3) in each iteration.
An \textit{outer iteration} of the \textbf{In-and-Out} algorithm is one iteration (corresponding to one value of $i \in \{0,1,\dots,T-1\}$) of the Steps 2--3 in Algorithm~\ref{Alg:In-and-Out}.
We call Step 2 the \textit{forward step}, and Step 3 the \textit{backward step}.

We note the \textbf{In-and-Out} algorithm is an implementation of the idealized \textit{Proximal Sampler} scheme (described in Section~\ref{Sec:ProxSampler}) via rejection sampling with a threshold $N$ on the number of trials in Step 3.
Without the threshold $N$ (equivalently, when $N = \infty$), the \textbf{In-and-Out} algorithm becomes exactly the \textit{Proximal Sampler} scheme.
However, as explained in~\cite[Section~3.3.1]{kook2026and}, without the threshold $N$, the expected number of trials in the rejection sampling in Step 3 (when $x_i \sim \pi$) is equal to infinity.
Therefore, following~\cite{kook2026and}, we introduce the threshold $N < \infty$ to ensure the expected number of trials in Step 3 is finite.
With the threshold $N < \infty$, the \textbf{In-and-Out} algorithm can fail when the rejection sampling in Step 3 exceeds $N$ trials in any iteration.
When $\X$ satisfies a volume growth condition, we can control the failure probability of the \textbf{In-and-Out} algorithm; see Lemma~\ref{Lem:SuccProbInAndOut} in Section~\ref{Sec:SuccProbInAndOut}.

When it succeeds for $T$ iterations, \textbf{In-and-Out} returns a point $x_T \sim \rho_T$ which is a random variable with a probability distribution $\rho_T$ supported on $\X$, i.e., $x_T \in \X$ almost surely.
In this case, the \textbf{In-and-Out} algorithm has a convergence guarantee (inherited from the \textit{Proximal Sampler} with a small bias introduced by the threshold $N$) in R\'enyi divergence between the output distribution $\rho_T$ and the target uniform distribution $\pi \propto \one_\X$, assuming that $\pi$ satisfies a Poincar\'e inequality; see Lemma~\ref{Lem:OuterConvergence} in Section~\ref{Sec:OuterConvergence}.

\subsection{Main result: Iteration complexity of In-and-Out with a warm start}

Our main result is the following guarantee for \textbf{In-and-Out} for sampling from $\pi \propto \one_\X$ under isoperimetry and volume growth.
Below, we assume $\beta \ge \frac{1}{n}$ without loss of generality; if $\beta$ is smaller than $\frac{1}{n}$, we can always replace it by $\frac{1}{n}$.
We provide the proof of Theorem~\ref{Thm:Nonconvex} in Appendix~\ref{App:NonconvexProof}.

\begin{theorem}\label{Thm:Nonconvex}
Let $\pi \propto \one_\X$ where $\X \subset \R^n$ is a compact body, $n \ge 2$.
Assume:
\begin{enumerate}
    \item $\X$ satisfies an $(\alpha,\beta)$-volume growth condition for some $\alpha \in [1,\infty)$ and $\beta \in [\frac{1}{n},\infty)$;
    \item $\pi$ satisfies a $\CPI$-Poincar\'e inequality for some $\CPI \in [1,\infty)$.
\end{enumerate}
Let $q \in [2,\infty)$, $\varepsilon \in (0, \frac{1}{2})$, and $M \in [1,\infty)$ be arbitrary. 
Suppose $x_0 \sim \rho_0$ is $M$-warm with respect to $\pi$.
Then with a suitable choice of parameters (see~\eqref{Eq:TPIDef} for $T$,~\eqref{Eq:hDef} for $h$,~\eqref{Eq:NDef} for $N$),  with probability at least $1-\varepsilon$, \textbf{In-and-Out} succeeds for $T$ iterations, 
and outputs $x_T \sim \rho_T$ satisfying $\sR_q(\rho_T \,\|\, \pi) \le \varepsilon$,
with the total number of trials over all $T$ iterations bounded in expectation by:
\begin{align*}
    \E\big[\# \text{ of trials in \textbf{In-and-Out}} \big] 
    &= \tilde O\left(q \, \CPI \, \alpha \, \beta^2 \, M \, n^3 \, \log^4 \frac{1}{\varepsilon} \right),
\end{align*}
where each trial requires one query to the membership oracle $\one_\X$ and $O(n)$ arithmetic operations.
\end{theorem}

The $\tilde O$ notation above hides logarithmic dependencies in the leading parameters $q$, $n$, $\CPI$, $\alpha$, $\beta$, $M$, and $\log \frac{1}{\varepsilon}$.
See the proof of Theorem~\ref{Thm:Nonconvex} in Appendix~\ref{App:NonconvexProof} for a precise bound (in equation~\eqref{Eq:NumTrialsBd1}) for the expected total number of trials of \textbf{In-and-Out}.
We also note the output guarantee in R\'enyi divergence $\sR_q$ of index $q\ge 2$ also implies guarantees in KL divergence and total variation distance, since $2\,\TV^2 \leq \KL\leq \sR_{2} \le \sR_q$ (see Section~\ref{Sec:Renyi}). 

We recall that convex bodies and star-shaped bodies have bounded Poincar\'e constants (see Section~\ref{Sec:Poincare}), and they also have bounded volume growth constants (see Section~\ref{Sec:VolGrowth}).
Thus, our result recovers polynomial-time sampling for convex bodies, albeit at a higher iteration complexity $\tilde O(n^3)$ compared to $\tilde O(n^2)$ from~\citep{kook2026and}.
Our result provides polynomial-time sampling (with polynomial dependence on $\log \varepsilon^{-1}$) for star-shaped bodies with guarantee in R\'enyi divergence; this strengthens the prior result of~\citep{CDV10} which has a polynomial dependence in $\varepsilon^{-1}$ and only provides guarantees in total variation distance. 
Perhaps most importantly, our result vastly generalizes the scope of polynomial-time uniform sampling to a large class of nonconvex bodies --- those satisfying the volume growth condition and whose uniform distributions satisfy a Poincar\'e inequality.
We note that our result still requires a warm start.
How to generate a warm start for a nonconvex body remains an open question, as all previous techniques for generating warm starts seem to heavily require convexity.

We remark on the differences compared to the convex case from~\cite{kook2026and}.
A key part in the analysis of \textbf{In-and-Out} is in showing that the forward step does not step too far away from $\X$, and that the backward step has a good chance of landing back in $\X$.
In the previous work of~\cite{kook2026and}, both these steps were analyzed by crucially assuming convexity; see for example~\cite[Lemma~13]{kook2026and}.
In our work, we analyze these steps without assuming convexity, only assuming that $\X$ satisfies a volume growth condition (Definition~\ref{Def:VolGrowth}).
We provide further discussion on the differences with the convex setting in a remark following Lemma~\ref{Lem:EscProb}.

\section{Preliminaries}
\label{Sec:Prelim}

\paragraph{Notation and definitions.}
Let $\R^n$ be the Euclidean space of dimension $n \ge 2$.
Let $\|x\| = \sqrt{\sum_{i=1}^n x_i^2}$ be the $\ell^2$-norm of a vector $x = (x_1,\dots,x_n) \in \R^n$.
Let $I_n$ denote the $n \times n$ identity matrix.
Let $\N(\mu, \Sigma)$ denote the Gaussian distribution with mean vector $\mu \in \R^n$ and covariance matrix $\Sigma \in \R^{n \times n}$.
In particular, $\N(0, I_n)$ is the standard Gaussian distribution in $\R^n$.

Given a set $\X \subseteq \R^n$, its \textit{interior} $\X^\circ$ is the set of points $x \in \X$ such that a sufficiently small ball around $x$ is still contained in $\X$.
Recall $\X \subseteq \R^n$ is \textit{open} if $\X = \X^\circ$.
Recall $\X \subseteq \R^n$ is \textit{closed} if $\X^\complement$ is open, where $\X^\complement = \R^n \setminus \X$ is its complement.
Recall a \textit{body} $\X$ is a set $\X \subseteq \R^n$ that is closed and has a non-empty interior: $\X^\circ \neq \emptyset$.
Recall a set $\X \subset \R^n$ is \textit{compact} if and only if it is closed and bounded.
Given a closed set $\X \subseteq \R^n$, recall the \textit{boundary} $\partial \X$ of $\X$ is the set of points in $\X$ that are not in the interior: $\partial \X = \X \setminus \X^\circ$.
Recall the \textit{volume} of a body $\X \subseteq \R^n$ is the integral over the body: $\Vol(\X) = \int_\X \, \dx$, where $\dx$ is the Lebesgue measure on $\R^n$.
The \textit{surface area} of a body $\X \subset \R^n$ is the integral over the boundary $\partial \X$: $\area(\partial \X) = \int_{\partial \X} \d \H^{n-1}(x)$, where $\d \H^{n-1}(x)$ is the $(n-1)$-dimensional Hausdorff measure.

Given distributions $\rho, \pi$ on $\R^n$, recall we say $\rho$ is \textit{absolutely continuous} with respect to $\pi$, denoted by $\rho \ll \pi$, if for any measurable set $A \subseteq \R^n$, $\pi(A) = 0$ implies $\rho(A) = 0$.
If a probability distribution $\pi$ is absolutely continuous with respect to the Lebesgue measure $\dx$, then we can write $\pi$ in terms of its probability density function (Radon-Nikodym derivative), which we also denote by $\pi \colon \R^n \to [0,\infty)$, with $\int_{\R^n} \pi(x) \, \dx = 1$.
If $\rho$ and $\pi$ are both absolutely continuous with respect to the Lebesgue measure $\dx$ and represented by their density functions, then $\rho \ll \pi$ means $\pi(x) = 0$ for some $x \in \R^n$ implies $\rho(x) = 0$.

\subsection{Geometry of the support set}
\label{Sec:Set}

Throughout, let $\X \subset \R^n$ be a compact body;
this means $\X$ is closed, bounded, and has a non-empty interior, so it has a finite volume $\Vol(\X) \in (0,+\infty)$.
Note we do not assume $\X$ is convex.
We assume $\X$ has a sufficiently regular surface, so $\X$ has a finite surface area $\area(\partial \X) \in (0,+\infty)$.

We assume we have access to a \textit{membership oracle} to $\X$, which is a function $\one_\X \colon \R^n \to \{0,1\}$ given by $\one_\X(x) = 1$ if $x \in \X$, and $\one_\X(x) = 0$ if $x \notin \X$.
We measure the complexity of our algorithm by the number of queries to the membership oracle $\one_\X$.

Our goal is to sample from the uniform probability distribution $\pi \propto \one_\X$ supported on $\X$.
Explicitly, $\pi$ has probability density function $\pi \colon \R^n \to [0,\infty)$ given by:
$$\pi(x) = \frac{1}{\Vol(\X)} \cdot \one_\X(x) \qquad \text{ for all } ~ x \in \R^n.$$

We also define the following notions for convenience.
Given a closed set $\Y \subseteq \R^n$, we define the \textit{distance function} $\dist(\cdot, \Y) \colon \R^n \to [0,\infty)$ of $x \in \R^n$ to $\Y$ by:
$$\dist(x,\Y) = \min_{y \in \Y} \|x-y\|.$$
Note that by definition, $\dist(x,\Y) = 0$ if and only if $x \in \Y$.

For $t \ge 0$, let $B_t \equiv B(0,t)$ be the $\ell_2$-ball of radius $t$ centered at $0 \in \R^n$:
$$B_t = \{x \in \R^n \colon \|x\| \le t\} = \{x \in \R^n \colon \dist(x,\{0\}) \le t\}.$$ 

Given $\X \subset \R^n$, for $t \ge 0$, we define the \textit{enlarged set} $\X_t \subset \R^n$ by:
\begin{align*}
    \X_t &= \X \oplus B_t = \{x \in \R^n \colon \dist(x,\X) \le t\}
\end{align*}
where $\oplus$ is the Minkowski sum between sets.
We note the following relation, which we can also take as the definition of surface area: 
$\area(\partial \X) = \lim_{t \to 0} \frac{1}{t} (\Vol(\X_t) - \Vol(\X)).$

Given a compact body $\X \subset \R^n$, we define the \textit{outer isoperimetry} of $\X$ to be the ratio of the surface area to volume:
$$\xi(\X) = \frac{\area(\partial \X)}{\Vol(\X)} \in (0,\infty).$$

\subsection{Isoperimetry: Poincar\'e inequality}
\label{Sec:Poincare}

Recall we say a probability distribution $\pi$ on $\R^n$ satisfies a \textbf{Poincar\'e inequality (PI)} with constant $\CPI(\pi) \in (0, \infty)$ if for all smooth functions $\phi \colon \R^n \to \R$, the following holds:
$$\Var_\pi(\phi) \le \CPI(\pi) \cdot \E_\pi[\|\nabla \phi\|^2]$$
where $\Var_\pi(\phi) = \E_\pi\left[(\phi - \E_\pi[\phi])^2 \right]$ is the variance of $\phi$ under $\pi$.

We recall that if $\pi$ is logconcave (i.e., $-\log \pi$ is a convex function on $\R^n$), then it satisfies a Poincar\'e inequality.
The Kannan-Lovasz-Simonovitz (KLS) conjecture~\citep{KLS95isop} states that if $\pi$ is logconcave, then $\CPI(\pi) = O(\|\Sigma\|_\op)$, where $\|\Sigma\|_\op$ is the largest eigenvalue of the covariance matrix  $\Sigma = \Cov_\pi(X)$ of $X \sim \pi$.
The current best result~\citep{klartag2023logarithmic} is $\CPI(\pi) = O(\|\Sigma\|_\op \cdot \log n)$. 

We recall that for an arbitrary distribution $\pi$, we have $\CPI(\pi) \lesssim \CCh(\pi)^{-2}$~\citep{Chee70lower}, where $\CCh(\pi)$ is the \textit{Cheeger constant} of $\pi$, defined by
$$\CCh(\pi) = \inf_{A \subset \R^n} \frac{\pi(\partial A)}{\min \left\{ \pi(A), \, \pi(A^\complement) \right\}}$$
where $\pi(\partial A)$ is the surface area of $A$ measured by $\pi$, defined by 
$\pi(\partial A) = \liminf_{t \to 0} \frac{\pi(A \oplus B_t) - \pi(A)}{t}.$

We also recall that the Poincar\'e inequality is stable under some types of perturbations of the distributions (including bounded perturbations of the density, and pushforward by a Lipschitz mapping), while these perturbations easily destroy logconcavity.
Therefore, the class of distributions satisfying Poincar\'e inequality is larger than the class of logconcave distributions.

In this work, we assume that the target uniform distribution $\pi \propto \one_\X$ satisfies a Poincar\'e inequality with some constant $\CPI \equiv \CPI(\pi) \in (0,+\infty)$.
Note while $\X$ may be nonconvex, this assumption implies $\X$ cannot be too bad; e.g., $\X$ must be connected and does not have a ``bottleneck''.
If $\X$ is a convex body with diameter $D > 0$, then $\CPI(\pi) = O(D^2)$.
We also recall from~\citep{CDV10} that if $\X$ is a star-shaped body with diameter $D > 0$ and the volume of its convex core is a fraction $\gamma \in (0,1]$ of the total volume, then $\CPI(\pi) = O(D^2/\gamma^2)$.

\subsection{Volume growth condition}
\label{Sec:VolGrowth}

We introduce the following key definition on the volume growth of the enlarged set $\X_t = \X \oplus B_t$, where recall $B_t \equiv B(0,t)$ is the $\ell_2$-ball of radius $t \ge 0$.

\begin{definition}\label{Def:VolGrowth}
    We say a compact body $\X \subset \R^n$ satisfies an \textbf{$(\alpha,\beta)$-volume growth} condition for some constants $\alpha \in [1, \infty)$ and $\beta \in (0, \infty)$ if for all $t > 0$:
    \begin{align}\label{Eq:VolGrowth}
        \frac{\Vol(\X_t)}{\Vol(\X)} \le \alpha \cdot (1 + t \beta)^n.
    \end{align}
\end{definition}

We recall that convex sets satisfy the volume growth condition where the constant is determined by the outer isoperimetry; see Lemma~\ref{Lem:VolGrowthConvex}.
We can show that a star-shaped body satisfies the volume growth condition with a constant inherited from the convex body; see Lemma~\ref{Lem:VolGrowthStarShaped}.
Moreover, we show that volume growth condition is preserved under some operations, including set union and set exclusion; see Lemma~\ref{Lem:VolGrowthUnion} and Lemma~\ref{Lem:VolGrowthExclusion}.
Therefore, the volume growth condition captures convex bodies and a wide class of nonconvex sets.
Later, we will see that the volume growth condition allows us to control the failure probability of the \textbf{In-and-Out} algorithm (Lemma~\ref{Lem:EscProb}).

\begin{remark}
    Any compact body $\X$ trivially satisfies the volume growth condition, but with a naive estimate of $\alpha$ that may be exponentially large such that it is not useful for our algorithmic purpose.
    Concretely, since $\X$ is a compact set with nonempty interior, it contains a ball of radius $r$ and is contained in a larger ball of radius $R$, so $B(x,r) \subseteq \X \subseteq B(x,R)$ for some $x \in \X$ and $R \ge r > 0$. 
    Then $\X_t \subseteq B(x,R+t)$ for all $t \ge 0$, and we can estimate:
    $$\frac{\Vol(\X_t)}{\Vol(\X)} 
    \le \frac{\Vol(B(x,R+t))}{\Vol(B(x,r))} 
    = \left(\frac{R}{r}\right)^n \cdot \left(1 + \frac{t}{R} \right)^n.$$
    This shows $\X$ satisfies the volume growth condition with $\alpha = (R/r)^n$ and $\beta = 1/R$.
    If we have additional structures on $\X$, such as being convex or star-shaped, or being a union or difference of sets satisfying volume growth condition, then we may obtain better estimates on $\alpha$ and $\beta$; see Lemmas~\ref{Lem:VolGrowthConvex},~\ref{Lem:VolGrowthStarShaped},~\ref{Lem:VolGrowthUnion}, and~\ref{Lem:VolGrowthExclusion} below.
\end{remark}

\begin{remark}
    The volume growth property has been studied in prior work.
    Notably, the Brunn-Minkowski theorem implies that convex bodies satisfy the volume growth condition with $\alpha = 1$; see Lemma~\ref{Lem:VolGrowthConvex}.
    Moreover,~\cite[Theorem~3.7]{fradelizi2014analogue} show that if $\X \subset \R^n$ is compact with a regularity condition,\footnote{Namely, $\epsilon \mapsto \Vol(\epsilon \X \oplus B_1)$ is twice-differentiable in a neighborhood of $0$, with a continuous second derivative at $0$.} then $t \mapsto \Vol(\X_t)$ is eventually $\frac{1}{n}$-concave, i.e., there exists $T_0 \in [0, \infty)$ such that $t \mapsto \Vol(\X_t)^{1/n}$ is a concave function for all $t \ge T_0$.
    This implies $\Vol(\X_t) \le \Vol(\X_{T_0}) \, (1 + (t-T_0) \beta)^n$ for all $t \ge T_0$, where $\beta = \frac{1}{n} \xi(\X_{T_0})$; this shows the enlarged body $\X_{T_0}$ satisfies the volume growth condition with $\alpha = 1$.
\end{remark}

\subsubsection{Volume growth for convex bodies}
\label{Sec:VolGrowthConvex}

When $\X$ is convex, it satisfies the volume growth condition with a constant bounded by the outer isoperimetry ratio $\xi(\X) = \frac{\area(\partial \X)}{\Vol(\X)}$;
this is a consequence of the Brunn-Minkowski theorem.
We provide the proof of Lemma~\ref{Lem:VolGrowthConvex} in  Appendix~\ref{Sec:ReviewConvex}.

\begin{lemma}\label{Lem:VolGrowthConvex}
If $\X \subset \R^n$ is a convex body, then it satisfies the $(1,\frac{1}{n} \xi(\X))$-volume growth condition.
\end{lemma}

We also recall that if $\X$ is convex and contains a ball of radius $r > 0$, then we can bound $\frac{1}{n} \xi(\X) \le \frac{1}{r}$; see e.g.~\cite[Lemma~2.1]{BNN13}.

\subsubsection{Volume growth for star-shaped bodies}

We recall from~\citep{CDV10} that a body $\X \subset \R^n$ is \textit{star-shaped} if it is a union of convex sets, all of which have a common (necessarily convex) intersection called the \textit{core} of $\X$.
We can bound the volume growth constant of star-shaped bodies similar to convex bodies.
We provide the proof of Lemma~\ref{Lem:VolGrowthStarShaped} in Appendix~\ref{Sec:VolGrowthStarShaped}.

\begin{lemma}\label{Lem:VolGrowthStarShaped}
    Let $\X \subset \R^n$ be a star-shaped body, so $\X = \bigcup_{i \in \I} \X^i$ where $\X^i$ is a convex body for each $i \in \I$ in a finite index set $\I$, and they share a common intersection $\Y = \X^i \cap \X^j \neq \emptyset$ for all $i \neq j$.
    Assume $\Y$ contains a ball of radius $r > 0$ centered at $0$, i.e., $B_r \subseteq \Y$.
    Then $\X$ satisfies the $(1, \frac{1}{r})$-volume growth condition.
\end{lemma}

\subsubsection{Volume growth under union}

We show that the volume growth condition is preserved under set union.
We provide the proof of Lemma~\ref{Lem:VolGrowthUnion} in Appendix~\ref{Sec:GrowthUnion}.

\begin{lemma}\label{Lem:VolGrowthUnion}
    Suppose $\X^i \subset \R^n$ is a compact body that satisfies the $(\alpha_i, \beta_i)$-volume growth condition for some $\alpha_i \in [1,\infty)$ and $\beta_i \in (0,\infty)$, for each $i \in \I$ in some finite index set $\I$.
    Let $q_\I$ be the probability distribution supported on $\I$ with density $q_\I(i) = \frac{\Vol(\X^i)}{\sum_{j \in \I} \Vol(\X^j)}$, for $i \in \I$.
    Then the union $\X = \bigcup_{i \in \I} \X^i$ satisfies the $(A,B)$-volume growth condition, where:
    \begin{align*}
        A &= \left(\max_{i \in \I} \alpha_i \right) \cdot \frac{\sum_{i \in \I} \Vol(\X^i)}{\Vol(\X)}, \\
        B &= \E_{I \sim q_\I}\left[ \, \beta_I^n \, \right]^{1/n} = \left( \frac{\sum_{i \in \I} \Vol(\X^i) \cdot \beta_i^n}{\sum_{j \in \I} \Vol(\X^j)} \right)^{\frac{1}{n}} \le \max_{i\in \I} \beta_i.
    \end{align*}
\end{lemma}

\subsubsection{Volume growth under set exclusion}

We also show that the volume growth condition is preserved under set exclusion.
We provide the proof of Lemma~\ref{Lem:VolGrowthExclusion} in Appendix~\ref{Sec:VolGrowthExclusion}.

\begin{lemma}\label{Lem:VolGrowthExclusion}
Let $\Y \subset \R^n$ be a compact body that satisfies the $(\alpha,\beta)$-volume growth condition for some $\alpha \in [1,\infty)$ and $\beta \in (0,\infty)$. 
Let $\X = \Y \setminus \Z$, where $\Z \subset \Y$ is an open set with $\Vol(\Z) < \Vol(\Y)$, and assume $\X$ is compact.
Then $\X$ satisfies the $(A, \beta)$-volume growth condition where $A = \alpha \cdot \frac{\Vol(\Y)}{\Vol(\X)}$.
\end{lemma}

\subsection{R\'enyi divergence and other statistical divergences}
\label{Sec:Renyi}

Let $\rho$ and $\pi$ be probability distributions on $\R^n$ absolutely continuous with respect to the Lebesgue measure, so we can represent them by their density functions.
Assume $\rho \ll \pi$.

The \textit{R\'enyi divergence} of order $q \in (1,\infty)$ between $\rho$ and $\pi$ is:
$$\sR_q(\rho \,\|\, \pi) = \frac{1}{q-1} \log \E_\pi\left[\left(\frac{\rho}{\pi}\right)^q\right].$$
Recall $\sR_q(\rho \,\|\, \pi) \ge 0$ for all $\rho$ and $\pi$, and $\sR_q(\rho \,\|\, \pi) = 0$ if and only if $\rho = \pi$.
We also recall the map $q \mapsto \sR_q(\rho \,\|\, \pi)$ is increasing.
The limit $q \to 1$ is the \textit{Kullback-Leibler (KL) divergence} or relative entropy:
$$\lim_{q \to 1} \sR_q(\rho \,\|\, \pi) = \KL(\rho \,\|\, \pi) = \E_\rho\left[\log \frac{\rho}{\pi} \right].$$
Therefore, $\KL(\rho \,\|\, \pi) \le \sR_q(\rho \,\|\, \pi)$ for all $q > 1$.
The \textit{total variation distance} between $\rho$ and $\pi$ is:
$$\TV(\rho, \pi) = \sup_{A \subseteq \R^n} |\rho(A) - \pi(A)|.$$ 
Recall by \textit{Pinsker's inequality}, we have
$2 \, \TV(\rho, \pi)^2 \le \KL(\rho \,\|\, \pi)$ for any $\rho \ll \pi$.
For a review of the properties above, see e.g.,~\cite{van2014renyi}.

Suppose $\rho$ and $\pi$ both have support $\X \subset \R^n$, and $\rho \ll \pi$.
We say that $\rho$ is \textbf{$M$-warm} with respect to $\pi$ for some $M \in [1, \infty)$ if 
$$\sup_{x \in \X} \frac{\rho(x)}{\pi(x)} \le M.$$
Note that if $\rho$ is $M$-warm with respect to $\pi$, then
$\KL(\rho\,\|\,\pi) \le \log M$ and $\sR_q(\rho\,\|\,\pi) \le \log M$ for all $q \in (1,\infty)$;
in particular, 
$$\sR_\infty(\rho\,\|\,\pi) := \lim_{q \to \infty} \sR_q(\rho\,\|\,\pi) \le \log M.$$

\subsection{Review: Outer convergence of In-and-Out under isoperimetry}
\label{Sec:OuterConvergence}

We recall that when it succeeds, \textbf{In-and-Out} has a rapid mixing guarantee to $\pi$ in R\'enyi divergence under Poincar\'e inequality; see Lemma~\ref{Lem:OuterConvergence} below.
This follows from combining the convergence guarantee of the ideal \textit{Proximal Sampler} algorithm under Poincar\'e inequality (see Lemma~\ref{Lem:PSOuterConvergence} in Section~\ref{Sec:ProxSampler} for a review), together with a control on the bias of the \textbf{In-and-Out} algorithm conditioned on the success event which was shown in~\cite[Lemma~5]{kook2026and}.
For completeness, we provide the proof of Lemma~\ref{Lem:OuterConvergence} in Appendix~\ref{Sec:OuterConvergenceProof}.

\begin{lemma}\label{Lem:OuterConvergence}
    Assume $\pi \propto \one_\X$ satisfies a Poincar\'e inequality with constant $\CPI \in [1,\infty)$.
    Let $q \in [2,\infty)$ and $h > 0$.
    Let $\rho_0$ be a probability distribution on $\X$ with $\rho_0 \ll \pi$ and $\sR_q(\rho_0 \,\|\, \pi) < \infty$.
    Let $T_0 := \max\left\{0, \left\lceil \frac{q \left(\sR_q(\rho_0 \,\|\, \pi)-1\right)}{\log\left(1 + \frac{h}{\CPI}\right)} \right\rceil \right\}$, and let $T \ge T_0$ be the desired number of iterations.
    Let $\Succ$ be the success event that \textbf{In-and-Out} (Algorithm~\ref{Alg:In-and-Out}) runs without failure for $T$ iterations.
    Assume $\Pr(\Succ) \ge 1-\eta$ for some $\eta \in [0,\frac{1}{2}]$.
    Then, conditioned on $\Succ$, the output $x_T \sim \rho_T$ of \textbf{In-and-Out} satisfies:
    $$\sR_q(\rho_T \,\|\, \pi) \le 
    \left(1 + \frac{h}{\CPI}\right)^{-\frac{1}{q} (T-T_0)} + 4\eta.$$
\end{lemma}

Lemma~\ref{Lem:OuterConvergence} gives the number of \textit{outer} iterations of the \textbf{In-and-Out} algorithm to achieve a desired error threshold.
Then it remains to control the complexity of implementing each outer iteration via rejection sampling; we can do this via the notion of volume growth condition.

\section{Analysis of the In-and-Out algorithm}
\label{Sec:AnalysisInAndOut}

We provide an analysis of the \textbf{In-and-Out} algorithm, focusing on the key steps to prove Theorem~\ref{Thm:Nonconvex}, and provide additional details in Appendix~\ref{App:ProofDetails}.
In Lemma~\ref{Lem:EscProb}, we bound the stationary escape probability under the volume growth condition.
We use this to bound the stationary failure probability in Lemma~\ref{Lem:SuccProbStat}, and bound the stationary expected number of trials in Lemma~\ref{Lem:NumTrialsStatCor}.
In Lemma~\ref{Lem:SuccProbInAndOut}, we prove inductively that warm start is maintained in each iteration, and use it to translate the bounds on failure probability and expected number of trials from stationarity to each iteration of the \textbf{In-and-Out} algorithm.
We combine these steps to prove Theorem~\ref{Thm:Nonconvex} in Appendix~\ref{App:NonconvexProof}.

\subsection{Bound on stationary escape probability under volume growth condition}
\label{Sec:BoundEscapeProb1}

Let $X \sim \pi \propto \one_\X$.
For $h > 0$, define the random variable
$Y_h := X + \sqrt{h} Z$
where $Z \sim \N(0,I_n)$ is independent, so 
$Y_h \sim \pi_h$ where
$\pi_h := \pi \ast \N(0, h I_n).$

In this section, we study the following ``stationary escape probability'' at scale $r > 0$ and $h > 0$:
$$\Pr(Y_h \notin \X_r) = \pi_h\left(\X_r^\complement \right).$$
(Here ``stationary'' means that $X \sim \pi$, so $Y_h \sim \pi_h$.)

For $m \in \mathbb{N}$, we define the \textit{Gaussian tail probability} $Q_m \colon [0,\infty) \to [0,1]$ by:
$$Q_m(r) := \Pr(\|Z\|\ge r)$$
where $Z \sim \N(0,I_m)$ is a standard Gaussian random variable in $\R^m$.

We show the following bound on the stationary escape probability under the volume growth condition. This will allow us to show that a typical step stays ``close'' to $\X$. We provide the proof of Lemma~\ref{Lem:EscProb} in Appendix~\ref{Sec:LemEscProbProof}.

\begin{lemma}\label{Lem:EscProb}
Assume $\X \subset \R^n$ satisfies the $(\alpha,\beta)$-volume growth condition for some $\alpha \in [1,\infty)$ and $\beta \in (0,\infty)$.
If $$0 < h \le \frac{1}{2n^3 \beta^2} \enspace ,$$
then for all $r \ge 0$:
\begin{align}\label{Eq:EscProb}
    \pi_h\left(\X_r^\complement \right) \le 
\alpha(n+1) \cdot Q_{2n}\left(\frac{r}{\sqrt{h}}\right).
\end{align}
\end{lemma}

\paragraph{Comparison with convex case:}
Let us compare this result with the convex case from~\cite{kook2026and}.
When $\X$ is convex and contains a ball of radius $1$,~\cite[Lemma~13]{kook2026and} shows that:
$\pi_h(\X_r^\complement) \le \exp(\frac{hn^2}{2}) \cdot Q_1\left(\frac{r-hn}{\sqrt{h}}\right) \le \exp(-\frac{r^2}{2h} + rn)$;
this involves the one-dimensional Gaussian tail $Q_1$, and thus has an exponential decay for all $r > 0$.
When $\X$ is nonconvex, our bound~\eqref{Eq:EscProb} in Lemma~\ref{Lem:EscProb} involves the $2n$-dimensional Gaussian tail $Q_{2n}\left(\frac{r}{\sqrt{h}}\right)$, which only provides an exponential decay after $r > \sqrt{hn}$.
This difference results in a smaller step size $h = \tilde O(n^{-3})$ in the nonconvex case, compared to $h = \tilde O(n^{-2})$ in the convex case~\citep{kook2026and};
this manifests as the $\tilde O(n^3)$ iteration complexity of \textbf{In-and-Out} in the nonconvex case (Theorem~\ref{Thm:Nonconvex}), compared to $\tilde O(n^2)$ in the convex case~\citep{kook2026and}.

Let us illustrate where the difference above comes from.
A key step of the proof of Lemma~\ref{Lem:EscProb} is to consider $y \in \X_r^\complement$, which lies at distance $u(y) > 0$ from $\X$, and we want to bound the probability $\Pr(y + \sqrt{h} Z \in \X)$ where $Z \sim \N(0,I_n)$.
\begin{itemize}
    \item When $\X$ is convex, there is a halfspace $H_y$ containing $y$ that is contained in $\X^\complement$ (see Figure~\ref{Fig:Sep}(a)), 
    so we can bound: $\Pr(y + \sqrt{h} Z \in \X) \le \Pr(y + \sqrt{h} Z \in H_y^\complement)$, which is equal to the one-dimensional Gaussian tail probability $Q_1\left(\frac{u(y)}{\sqrt{h}}\right)$.
    
    \item When $\X$ is nonconvex, such a halfspace containment may not hold; however, a weaker ball containment still holds.
    Note the ball of radius $u(y) = \dist(y,\X)$ centered at $y$ does not intersect $\X$ (see Figure~\ref{Fig:Sep}(b)):
    $B(y,u(y)) \subseteq \X^\complement$,
    so
    $\X \subseteq B(y,u(y))^\complement$.
    Then we can bound: $\Pr(y + \sqrt{h} Z \in \X) \le \Pr(y + \sqrt{h} Z \in B(y,u(y))^\complement)$, which is equal to the $n$-dimensional Gaussian tail probability $Q_{n}\left(\frac{u(y)}{\sqrt{h}}\right)$.
    Subsequent steps of the proof result in a bound involving the $2n$-dimensional Gaussian tail probability $Q_{2n}$.
    See Appendix~\ref{Sec:LemEscProbProof} for details.
\end{itemize}

\begin{figure}[t!]
    \centering
    \subfloat[$\X$ convex: There is a halfspace containing $y$ that does not intersect $\X$.]{\includegraphics[width=.4\textwidth]{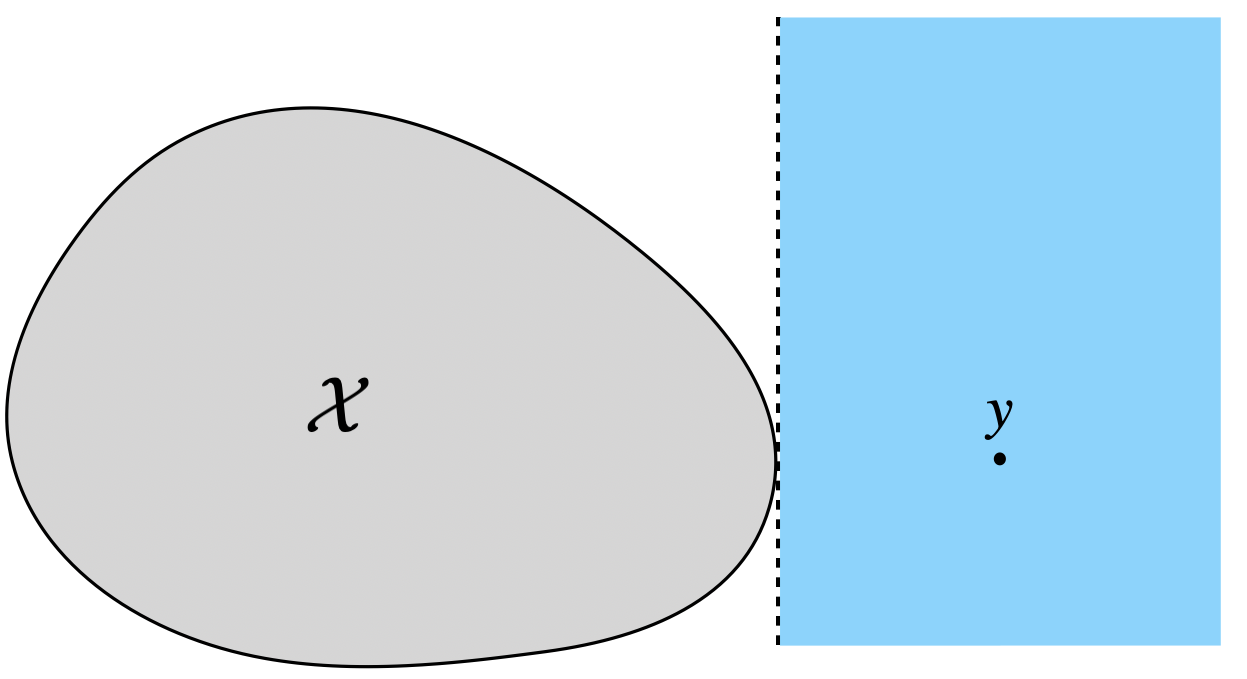}}
    \qquad
    \subfloat[$\X$ nonconvex: There is a ball containing $y$ that does not intersect $\X$.]{\includegraphics[width=.4\textwidth]{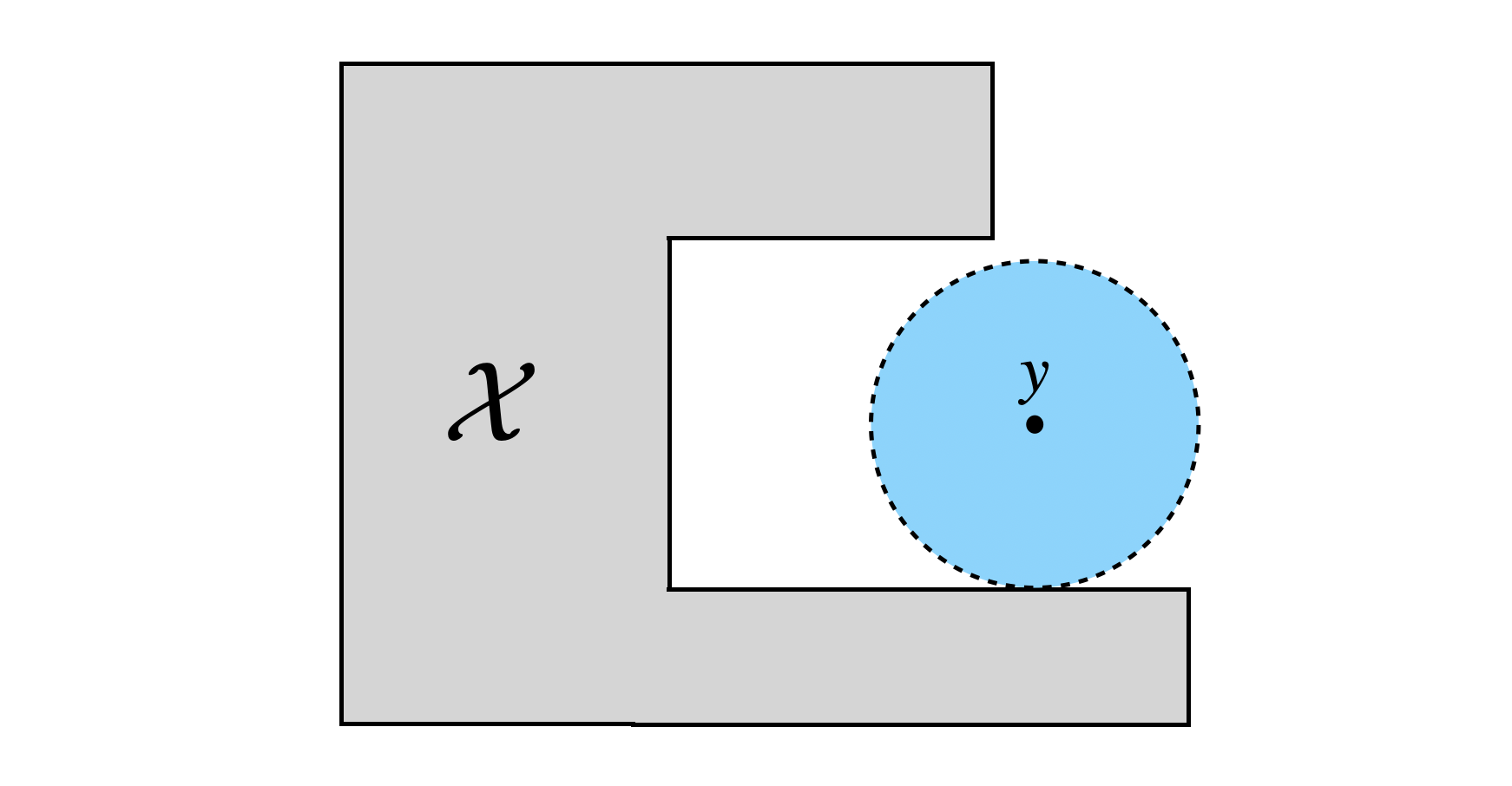}}
    \caption{Difference of containment guarantee in the (a) convex case and (b) nonconvex case.}
    \label{Fig:Sep}
\end{figure}

\subsection{Bound on stationary failure probability}
\label{Sec:SuccessStationary}

For $h > 0$, let $\ell_h \colon \R^n \to [0,1]$ denote the \textit{local conductance}~\citep{kook2026and}:
\begin{align}\label{Eq:LocalConductance}
    \ell_h(y) = \int_\X \N_h(y-x) \, \dx = \Pr_{Z \sim \N(0, I_n)}\left(y+\sqrt{h}Z \in \X \right).
\end{align}

For $h > 0$ and $N \in \bN$, we define the failure probability under stationarity (where ``stationarity'' means the expectation is under $\pi_h$) to be:
\begin{align}\label{Eq:FailProbStat}
    F(h,N) := \E_{Y \sim \pi_h}[(1-\ell_h(Y))^N].
\end{align}
We have the following bound.
We provide the proof of Lemma~\ref{Lem:SuccProbStat} in Appendix~\ref{App:SuccessStationary}.

\begin{lemma}\label{Lem:SuccProbStat}
    Let $\X \subset \R^n$ be a compact body which satisfies the $(\alpha,\beta)$-volume growth condition.
    Assume:
    \begin{align*}
        h &\le \frac{1}{2n^2 \beta^2 \cdot \max\{n, \, \log((n+1) \alpha S) \}} \\
        N &= 8\alpha S \log S
    \end{align*}
    where $S \ge 3$ is arbitrary.
    Then we have: 
    \begin{align*}
        F(h,N) \le \frac{3}{S}.
    \end{align*}
\end{lemma}

\subsection{Bound on the expected number of trials under stationarity}
\label{Sec:NumTrialsStationary}

For $h > 0$, $N \in \bN$, and $y \in \R^n$, let $M_{h,N}(y)$ be a random variable that denotes the number of trials in the rejection sampling for implementing one backward step of the \textbf{In-and-Out} algorithm with step size $h$ starting from $y$ and limited to $N$ trials, where each trial is independent with success probability given by the local conductance $\ell_h(y)$.
By definition, we can write
\begin{align}\label{Eq:EqDistTrials}
    M_{h,N}(y) \stackrel{d}{=} \min\{G_h(y), N\}
\end{align}
where $\stackrel{d}{=}$ denotes equality in distribution, and $G_h(y) \in \bN$ is a geometric random variable with success probability $\ell_h(y)$.

Define the \textbf{expected number of trials} in the rejection sampling (for implementing one backward step) starting from $y$:
\begin{align}\label{Eq:NumTrialsStat}
    \mu_{h,N}(y) = \E[M_{h,N}(y)]
\end{align} 
where the expectation is over the geometric random variable $G_h(y)$.
Note for all $h > 0$, $N \in \bN$, and $y \in \R^n$:
\begin{align}\label{Eq:PtwiseBound}
    \mu_{h,N}(y) \le \min \left\{\frac{1}{\ell_h(y)}, N\right\}.
\end{align}
(Indeed, since $M_{h,N}(y) \le N$, we have $\mu_{h,N}(y) = \E[M_{h,N}(y)] \le N$.
And since $M_{h,N}(y) \le G_h(y)$, we also have $\mu_{h,N}(y) = \E[M_{h,N}(y)] \le \E[G_h(y)] = 1/\ell_h(y)$.)

We have the following bound.
We provide the proof of Lemma~\ref{Lem:NumTrialsStatCor} in Appendix~\ref{App:NumTrialsStationary}.

\begin{lemma}\label{Lem:NumTrialsStatCor}
    Let $\X \subset \R^n$ be a compact body that satisfies the $(\alpha,\beta)$-volume growth condition for some $\alpha \in [1,\infty)$ and $\beta \in (0,\infty)$.
    Assume:
    \begin{align*}
        h &\le \frac{1}{2n^2 \beta^2 \cdot \max\{n, \, \log((n+1)\alpha S)\}} \\
        N &= 8\alpha S \log S
    \end{align*}
    where $S \ge 3$ is arbitrary.
    Then:
    $$\E_{\pi_h}[\mu_{h,N}] \le 16 \alpha \log S.$$
\end{lemma}

\subsection{Per-iteration success probability and expected number of trials}
\label{Sec:SuccProbInAndOut}

Recall an iteration of \textbf{In-and-Out} is one forward step and one backward step which involves up to $N$ rejection sampling trials.
We have the following guarantee on the success probability and the expected number of trials of the rejection sampling step in the \textbf{In-and-Out} algorithm.
We provide the proof of Lemma~\ref{Lem:SuccProbInAndOut} in Appendix~\ref{Sec:SuccProbInAndOutProof}.

\begin{lemma}\label{Lem:SuccProbInAndOut}
    Let $\pi \propto \one_\X$ where $\X \subset \R^n$ satisfies the $(\alpha,\beta)$-volume growth condition for some $\alpha \in [1,\infty)$ and $\beta \in (0,\infty)$. 
    Let $x_0 \sim \rho_0$ where $\rho_0$ is $M$-warm with respect to $\pi$ for some $M \in [1,\infty)$.
    Assume:
    \begin{align*}
        h &\le \frac{1}{2\beta^2 n^3 \cdot \max\{1, \, \frac{1}{n} \log((n+1)\alpha S)\}} \\
        N &= 8\alpha S \log S
    \end{align*}
    where $S \ge 3$ is arbitrary.
    Assume \textbf{In-and-Out} succeeds for $i \ge 0$ iterations (i.e., conditioned on the event ``reach iteration $i$''), and is currently at $x_i \sim \rho_i$ which is $M$-warm with respect to $\pi$.
    Then the following hold:
    \begin{enumerate}
        \item The next iteration (to compute $x_{i+1}$) succeeds with probability:
        $$\Pr_{x_i \sim \rho_i} \left(\text{next iteration succeeds} \mid \text{reach iteration $i$}\right) \ge 1-\frac{3M}{S}.$$
        \item The expected number of trials in the next iteration is
        $$\E_{x_i \sim \rho_i} \left[\text{$\#$ trials until success in next iteration} \mid \text{reach iteration $i$} \right] \le 16 M\alpha  \log S.$$
        \item Upon accepting $x_{i+1} \sim \rho_{i+1}$, the next distribution $\rho_{i+1}$ is $M$-warm with respect to $\pi$.
    \end{enumerate}
\end{lemma}

The lemma above allows us to induct over the iterations and bound the failure probability of \textbf{In-and-Out} over $T$ iterations.
We use it to prove Theorem~\ref{Thm:Nonconvex} in Appendix~\ref{App:NonconvexProof}.

\section{Discussion and future work}

In this paper, we have provided an efficient algorithm (\textbf{In-and-Out}) to sample from a large family of compact, nonconvex sets.
Our result significantly advances the frontier of provably efficient uniform sampling; such results were previously known for the special convex bodies and star-shaped bodies.
We show it for nonconvex sets satisfying isoperimetry (Poincar\'e inequality) and a volume growth condition; these conditions are preserved under natural set operations, and thereby capture a much wider class beyond convexity.

The \textbf{In-and-Out} algorithm is an implementation of the ideal Proximal Sampler scheme, which can be viewed as an approximate proximal discretization of the Langevin dynamics, and explains why isoperimetry in the form of Poincar\'e inequality is a natural condition for fast mixing.
The difficulty is in controlling the complexity of implementing the Proximal Sampler scheme using rejection sampling.
We showed how to control the failure probability under the volume growth condition.
We conjecture that a Poincar\'e inequality and a bounded outer isoperimetry $\xi(\X) = \area(\partial \X)/\Vol(\X)$ (which is even weaker than volume growth) are sufficient for efficient sampling.
Our result in this paper assumes a warm start initialization.
How to generate a warm start efficiently for a nonconvex set is an open problem.

It would be interesting to study if other algorithms, such as the Ball Walk or the Metropolis Random Walk, also work for sampling from this class of nonconvex sets; all existing proofs still assume convexity/star-shapedness.
Beyond uniform sampling, it would be interesting to study how to efficiently sample more general distributions under isoperimetry and \textit{without} smoothness, e.g.\ extending the works of~\cite{AK91sampling,frieze1999,LV06,kookV2025} beyond logconcavity.

\appendix

\section{Additional discussion}

\subsection{A review of the Proximal Sampler}
\label{Sec:ProxSampler}

The ideal \textit{Proximal Sampler} algorithm is the following (Algorithm~\ref{Alg:ProxSamp}).
Here \textit{ideal} means it assumes we can implement both the forward step \ref{Eq:AlgFwdProx} and backward step \ref{Eq:AlgBwdProx} in Algorithm~\ref{Alg:ProxSamp}.
If we implement the backward step~\ref{Eq:AlgBwdProx} via rejection sampling with a threshold on the number of trials, then we obtain the \textbf{In-and-Out} algorithm (Algorithm~\ref{Alg:In-and-Out}).
Below, $\N(y, h I_{n}) \cdot \one_{\X}$ denotes the Gaussian distribution $\N(y, h I_{n})$ restricted to $\X$.

\begin{algorithm}[H]
\hspace*{\algorithmicindent} 
\begin{algorithmic}[1] 
\caption{: \textbf{Proximal Sampler} for sampling from $\pi \propto \one_\X$ \\
\textbf{Input:} Initial point $x_0 \sim \rho_0 \in \P(\X)$, step size $h > 0$, number of steps $T \in \mathbb{N}$ \\
\textbf{Output:} $x_T \sim \rho_{T} \in \P(\X)$
} 
\label{Alg:ProxSamp}

\FOR{$i=0,\dotsc,T-1$}

\STATE Sample $y_{i}\sim
\N(x_{i},hI_{n})$
\label{Eq:AlgFwdProx}

\STATE Sample $x_{i+1} \sim \N(y_{i}, h I_{n}) \cdot \one_{\X}$
\label{Eq:AlgBwdProx}

\ENDFOR
\end{algorithmic}
\end{algorithm}

\paragraph{Proximal Sampler as Gibbs sampling.}
As described in~\citep{LST21structured}, we can derive the \textit{Proximal Sampler} algorithm as an alternating Gibbs sampling algorithm from a joint target distribution $\pi^{XY} \in \P(\R^{2n})$ with density function:
\begin{align*}
    \pi^{XY}(x,y) \propto \one_\X(x) \cdot \exp\left(-\frac{1}{2h} \|y-x\|^2 \right).
\end{align*}
Note the $X$-marginal of $\pi^{XY}$ is $\pi \propto \one_\X$.
If we can sample from $\pi^{XY}$, then we can return the $X$-component as a sample from $\pi$.
We apply \textit{Gibbs sampling} to sample from $\pi^{XY}$, which results in the following reversible Markov chain $(x_i,y_i) \mapsto (x_{i+1}, y_{i+1})$ via the alternating update:
\begin{align*}
    x_{i+1} \mid y_i \sim \pi^{X \mid Y=y_i} = \N(y_i, hI_n) \cdot \one_\X \\
    y_{i+1} \mid x_{i+1} \sim \pi^{Y \mid X = x_{i+1}} = \N(x_{i+1}, hI_n)
\end{align*}
which is precisely the \textit{Proximal Sampler} (Algorithm~\ref{Alg:ProxSamp}).
See also~\cite{CCSW22improved} for further optimization interpretation of the \textit{Proximal Sampler} as an approximate proximal discretization of the continuous-time Langevin dynamics.

\paragraph{A review on the convergence guarantee of the Proximal Sampler.}
We recall the following result on the convergence guarantee of the \textit{Proximal Sampler} under Poincar\'e inequality.
This result was shown in~\citep[Theorem~4]{CCSW22improved} for target distribution $\pi$ with full support on $\R^n$, and was extended to the case of uniform distribution $\pi \propto \one_\X$ in~\citep[Theorem~2]{kook2026and}.
We refer the reader to~\citep[Theorem~2]{kook2026and} for the proof.

\begin{lemma}[{\citep[Theorem~2]{kook2026and}}]\label{Lem:PSOuterConvergence}
    Assume $\pi \propto \one_\X$ satisfies a Poincar\'e inequality with constant $\CPI \in [1,\infty)$.
    Let $q \in [2,\infty)$.
    Let $\rho_0$ be a probability distribution on $\X$ with $\rho_0 \ll \pi$ and $\sR_q(\rho_0 \,\|\, \pi) < \infty$.
    Along the \textit{Proximal Sampler} algorithm from $x_0 \sim \rho_0$ with step size $h > 0$ for $T \in \mathbb{N}$ iterations, the output $x_T \sim \rho_T$ satisfies:
    $$\sR_q(\rho_T \,\|\, \pi) \le 
    \begin{cases}
        \sR_q(\rho_0 \,\|\, \pi) - \frac{T}{q} \log\left(1 + \frac{h}{\CPI}\right) ~~ & \text{ if } T \le \frac{q \left(\sR_q(\rho_0 \,\|\, \pi)-1\right)}{\log\left(1 + \frac{h}{\CPI}\right)} \\
        \left(1 + \frac{h}{\CPI}\right)^{-\frac{1}{q} (T-T_0)} & \text{ if } T \ge T_0 := \max\left\{0, \left\lceil \frac{q \left(\sR_q(\rho_0 \,\|\, \pi)-1\right)}{\log\left(1 + \frac{h}{\CPI}\right)} \right\rceil \right\}.
    \end{cases}$$
\end{lemma}

\subsection{The analogy with optimization}

The results for sampling logconcave densities can be viewed as parallels to classical results for optimization, where again, there are polynomial-time algorithms for convex optimization (when the objective function to be minimized and the feasible region of solutions are both convex)~\citep{GLS93geometric,Vai,lee2018efficient}. Since a logconcave density $\pi$ has a convex potential, i.e., $\pi(x) \propto e^{-f(x)}$ for some convex function $f$, this family is the natural analog of convex functions in optimization. 

\paragraph{Beyond convexity/logconcavity.}
For optimization, $\min_{x \in \X} f(x)$, there are mild deviations from convexity that still allow efficient algorithms; notably, when the objective function satisfies a gradient-domination condition (Polyak-\L{}ojaciewicz/PL inequality), simple algorithms such as the proximal gradient method or gradient descent (under smoothness) have exponential convergence guarantees in objective function to the minimum value~\citep{karimi2016linear}.
However, even under the PL condition, it is crucial that any local minimum is also a global minimum. 
On the other hand, if the function $f$ is allowed to have spurious local (non-global) minima, then the task of finding the global minimizer is NP-hard in the worst case.

For sampling, there are more significant extensions beyond logconcavity that still allow for efficient algorithms. 
We mention two directions of existing results for efficient sampling beyond log-concavity.
First, star-shaped bodies can be efficiently sampled with complexity polynomial in the dimension and the inverse fraction of volume taken up by the convex core (the non-empty subset of the star-shaped body that can ``see'' all points in the body)~\citep{CDV10}; as shown in~\citep{CDV10}, linear optimization over star-shaped bodies is hard, even to solve approximately, and even when the convex core takes up a constant fraction of the star-shaped body. 
Second, distributions satisfying a Poincar\'e or log-Sobolev inequality under \textit{smoothness}, i.e., with a bound on the Lipschitz constant of the logarithm of the target density~\citep{VW23rapid,CCSW22improved}; however, the iteration complexity scales with the smoothness constant, and in any case it does not apply to our setting of uniform on $\X$.
The work of~\cite{abbasi2017hit} studies the Hit-and-Run algorithm for sampling the uniform distribution over a nonconvex set under the assumption that it is the image of a convex set under a bi-Lipschitz measure preserving map, with a stronger oracle assumption beyond membership oracle, and with an iteration complexity that depends on the smoothness and curvature of the set.

Despite only such modest deviations from convexity being known to be tractable, the existing results above demonstrate that local minima being global minima or unimodality is not essential for sampling. In fact, intuition suggests that sampling a substantially more general class of distributions should be tractable. For optimization, the optimal solution can be hidden in a very small part of the feasible region of a slightly nonconvex domain, such a bottleneck appears unlikely for sampling --- small parts of the domain can be effectively ignored as they take up small measure. Moreover, assuming isoperimetry suggests that all parts are ``reachable''. 
This leads us to our first motivating question as described in Section~\ref{Sec:Intro}, and which we answer in this paper via the additional notion of volume growth condition.

\section{Proofs for the Analysis of In-and-Out}
\label{App:ProofDetails}

\subsection{Proof of Lemma~\ref{Lem:OuterConvergence}}
\label{Sec:OuterConvergenceProof}

\begin{replemma}{Lem:OuterConvergence}
    Assume $\pi \propto \one_\X$ satisfies a Poincar\'e inequality with constant $\CPI \in [1,\infty)$.
    Let $q \in [2,\infty)$ and $h > 0$.
    Let $\rho_0$ be a probability distribution on $\X$ with $\rho_0 \ll \pi$ and $\sR_q(\rho_0 \,\|\, \pi) < \infty$.
    Let $T_0 := \max\left\{0, \left\lceil \frac{q \left(\sR_q(\rho_0 \,\|\, \pi)-1\right)}{\log\left(1 + \frac{h}{\CPI}\right)} \right\rceil \right\}$, and let $T \ge T_0$ be the desired number of iterations.
    Let $\Succ$ be the success event that \textbf{In-and-Out} (Algorithm~\ref{Alg:In-and-Out}) runs without failure for $T$ iterations.
    Assume $\Pr(\Succ) \ge 1-\eta$ for some $\eta \in [0,\frac{1}{2}]$.
    Then, conditioned on $\Succ$, the output $x_T \sim \rho_T$ of \textbf{In-and-Out} satisfies:
    $$\sR_q(\rho_T \,\|\, \pi) \le 
    \left(1 + \frac{h}{\CPI}\right)^{-\frac{1}{q} (T-T_0)} + 4\eta.$$
\end{replemma}
\begin{proof}
    Let $\tilde \rho_T$ be the output of the ideal \textit{Proximal Sampler} algorithm with step size $h$ for $T$ iterations, starting from the same initialization $x_0 \sim \rho_0$ as \textbf{In-and-Out}.
    By Lemma~\ref{Lem:PSOuterConvergence}, we have the guarantee 
    \begin{align}\label{Eq:ProxSampConv}
        \sR_q(\tilde \rho_T \,\|\, \pi) \le \left(1 + \frac{h}{\CPI}\right)^{-\frac{1}{q} (T-T_0)}.
    \end{align}
    By~\citep[Lemma~5]{kook2026and}, the output $x_T \sim \rho_T$ of \textbf{In-and-Out} conditioned on the success event $\Succ$ satisfies:
    \begin{align*}
        \sR_q(\rho_T \,\|\, \pi) 
        \,\le\, \sR_q(\tilde \rho_T \,\|\, \pi) + \frac{q}{q-1} \log \frac{1}{1-\eta}
        \,\stackrel{\eqref{Eq:ProxSampConv}}{\le}\, \left(1 + \frac{h}{\CPI}\right)^{-\frac{1}{q} (T-T_0)} + 4\eta
    \end{align*}
    where the last inequality follows since $q \ge 2$ so $\frac{q}{q-1} \le 2$, and $\eta \le \frac{1}{2}$ so $\log \frac{1}{1-\eta} \le 2\eta$.
\end{proof}

\subsection{Proof of Lemma~\ref{Lem:EscProb}}
\label{Sec:LemEscProbProof}

We recall a few definitions.
For $h > 0$, let $\N_h \colon \R^n \to \R$ denote the probability density function of the Gaussian distribution $\N(0,hI_n)$:
$$\N_h(y) = (2\pi h)^{-\frac{n}{2}} \exp\left(-\frac{\|y\|^2}{2h}\right).$$
For $h > 0$, recall $\ell_h \colon \R^n \to [0,1]$ denotes the \textit{local conductance}:
\begin{align*}
    \ell_h(y) = \int_\X \N_h(y-x) \, \dx = \Pr_{Z \sim \N(0, I_n)}\left(y+\sqrt{h}Z \in \X \right).
\end{align*}
For $h > 0$, let
$\pi_h = \pi \ast \N(0,hI_n)$,
so $\pi_h$ has density function at any point $y \in \R^n$:
\begin{align}\label{Eq:FormulaDensity}
    \pi_h(y) = \frac{1}{\Vol(\X)} \int_\X \N_h(y-x) \, \dx = \frac{\ell_h(y)}{\Vol(\X)}.
\end{align}

\begin{replemma}{Lem:EscProb}
Assume $\X \subset \R^n$ satisfies the $(\alpha,\beta)$-volume growth condition for some $\alpha \in [1,\infty)$ and $\beta \in (0,\infty)$.
If $$0 < h \le \frac{1}{2n^3 \beta^2} \enspace ,$$
then for all $r \ge 0$:
$$\pi_h\left(\X_r^\complement \right) \le 
\alpha(n+1) \cdot Q_{2n}\left(\frac{r}{\sqrt{h}}\right).$$
\end{replemma}
\begin{proof}
    Let $\pi_h = \pi \ast \mathcal{N}(0, h I_n)$, so it has density function (see~\eqref{Eq:FormulaDensity}):
    $$\pi_h(y) = \frac{1}{\Vol(\X)} \int_\X \N_h(y-x) \, dx = \frac{\Pr_{Z \sim \N(0,I_n)}(y + \sqrt{h} Z \in \X)}{\Vol(\X)}.$$
    Recall the distance function $u(y) := \dist(y,\X) = \min_{x \in \X} \|x-y\|$.    
    If $y \in \X^\complement$, then $u(y) > 0$ and 
    $B(y,u(y)) \subseteq \X^\complement$, where $B(y,u(y))$ is the $\ell_2$-ball of radius $u(y)$ centered at $y$.
    Therefore:
    \begin{align}\label{Eq:BallContainment}
        \X \subseteq B(y,u(y))^\complement.
    \end{align}
    Then for $y \in \X^\complement$:
    \begin{align*}
        \Pr(y+\sqrt{h}Z \in \X) 
        &\le
        \Pr\left(y+\sqrt{h}Z \in B(y,u(y))^\complement\right) 
        = \Pr\left(\|Z\| \ge \frac{u(y)}{\sqrt{h}}\right) 
        = Q_n\left(\frac{u(y)}{\sqrt{h}}\right).
    \end{align*}
    This means for any $y \in \X^\complement$ we have the bound:
    \begin{align}\label{Eq:BoundProbSmooth}
        \pi_h(y) &= \frac{\Pr(y+\sqrt{h}Z \in \X)}{\Vol(\X)} 
        \le \frac{1}{\Vol(\X)} Q_n\left(\frac{u(y)}{\sqrt{h}}\right).
    \end{align}
    Fix $r_0 \ge 0$ (this is $r$ in the statement of the lemma);
    we want to bound $\pi_h\left(\X_{r_0}^\complement \right)$.
    By the co-area formula, we can write:
    \begin{align}
    \allowdisplaybreaks
        \pi_h\left(\X_{r_0}^\complement \right)
        &= \int_{\X_{r_0}^\complement}  \pi_h(y) \, \dy \notag \\
        &\stackrel{\eqref{Eq:BoundProbSmooth}}{\le} \frac{1}{\Vol(\X)} \int_{\X_{r_0}^\complement} Q_n\left(\frac{u(y)}{\sqrt{h}}\right) \, \dy  \notag \\
        &= \frac{1}{\Vol(\X)} \int_{r_0}^\infty \int_{\{y \in \R^n \colon u(y) = r\}} Q_n\left(\frac{r}{\sqrt{h}}\right) \, \d\mathcal{H}^{n-1}(y) \, \dr \notag \\
        &= \frac{1}{\Vol(\X)} \int_{r_0}^\infty Q_n\left(\frac{r}{\sqrt{h}}\right) \cdot \underbrace{\int_{\{y \in \R^n \colon u(y) = r\}} \, \d\mathcal{H}^{n-1}(y)}_{= A(r)} \, \dr \notag \\
        &= \frac{1}{\Vol(\X)} \int_{r_0}^\infty Q_n\left(\frac{r}{\sqrt{h}}\right) \, A(r) \, \dr. \label{Eq:Calc01}
    \end{align}
    In the third line above, we use the co-area formula to write the integration in terms of the level sets of the distance function $u(y)$.
    We also use the fact that the distance function $u(y) = \dist(y,\X)$ satisfies $\|\nabla u(y)\| = 1$ for almost every $y \in \X^\complement$, which was shown in~\cite[Proof of Lemma~5.4]{BNN13} without assuming convexity; see also Lemma~\ref{Lem:DistanceFunctionLipschitz} for a self-contained proof.
    In the last line above, we define $A(r)$ to be the surface area of $\X_r$ which is equal to the integral of the $(n-1)$-dimensional Hausdorff measure $\d\mathcal{H}^{n-1}$ on the level set of constant distance:
    $$A(r) := \int_{\{y \in \R^n \colon u(y) = r\}} d\mathcal{H}^{n-1}(y) = \area(\partial(\X_r)).$$

    Next, similar to the convex case~\citep{kook2026and}, we do integration by parts.    
    Let $Z \sim \N(0, I_n)$ in $\R^n$.
    Recall $Y = \|Z\|$ has the chi distribution with density function $\gamma_n(r)$ at $r \ge 0$ given by:
    $$\gamma_n(r) = \frac{1}{N_n} r^{n-1} \exp\left(-\frac{r^2}{2}\right)$$
    where $N_n$ is the normalizing constant
    $$N_n := 2^{\frac{n}{2}-1} \Gamma\left(\frac{n}{2}\right)$$
    and recall the Gamma function $\Gamma(m) = \int_0^\infty e^{-t} \, t^{m-1} \dt$ for $m > 0$.
    Then by definition of $Q_n$,
    $$Q_n'(r) = \frac{\d}{\dr}Q_n(r) = -\gamma_n(r).$$
    
    For $h > 0$, define $G_h \colon [0,\infty) \to [0,1]$ by, for all $r \ge 0$:
    $$G_h(r) := \Pr\left(\|Z\| \ge \frac{r}{\sqrt{h}} \right) = Q_n\left(\frac{r}{\sqrt{h}}\right).$$
    We note the property:
    $$G_h'(r) = \frac{\d}{\dr}G_h(r) = \frac{1}{\sqrt{h}} Q_n'\left(\frac{r}{\sqrt{h}}\right)
    = - \frac{1}{\sqrt{h}} \gamma_n\left(\frac{r}{\sqrt{h}}\right) = -\frac{1}{h^{n/2} N_n} r^{n-1} \exp\left(-\frac{r^2}{2h}\right).$$

    Define the volume
    $\V(r) = \Vol(\X_r)$,
    and recall the relation:
    $$\V'(r) = \frac{\d}{\dr} \V(r) = \area(\partial (\X_r)) = A(r).$$

    Then by integration by parts:
    \begin{align}
        \int_{r_0}^\infty Q_n\left(\frac{r}{\sqrt{h}}\right) A(r) \, \dr
        &= \int_{r_0}^\infty G_h(r) \, \V'(r) \, \dr \notag \\
        &= \left(G_h(r) \, \V(r) \right)\Bigg|_{r=r_0}^\infty - \int_{r_0}^\infty G_h'(r) \, \V(r) \, \dr \notag \\
        &= \underbrace{\lim_{r \to \infty} G_h(r) \, \V(r)}_{=0} - \underbrace{G_h(r_0) \V(r_0)}_{\ge 0} 
        + \frac{1}{h^{n/2}N_n} \int_{r_0}^\infty  r^{n-1} \exp\left(-\frac{r^2}{2h}\right) \, \V(r) \, \dr \notag \\
        &\le \frac{1}{h^{n/2}N_n} \int_{r_0}^\infty  r^{n-1} \exp\left(-\frac{r^2}{2h}\right) \, \V(r) \, \dr \notag \\
        &= \frac{1}{N_n} \int_{r_0/\sqrt{h}}^\infty  u^{n-1} \exp\left(-\frac{u^2}{2}\right) \, \V(u \sqrt{h}) \, \du. \label{Eq:Calc02}
    \end{align}
    where the last inequality is because $\lim_{r \to \infty} G_h(r) \, \V(r) = 0$ (since $\V(r)$ has a polynomial growth by the volume growth assumption, and $G_h(r)$ has an exponential decay), and also because trivially $G_h(r_0) \V(r_0) \ge 0$.
    In the last equality we use change of variable
    $u = \frac{r}{\sqrt{h}}$, so $\dr = \sqrt{h} \, \du$ and note the $h^{n/2}$ term disappears in the last line.

    By the volume growth condition, we have for all $u,h > 0$:
    $$\frac{\V(u \sqrt{h})}{\V(0)} 
    \le \alpha \left(1+u\sqrt{h} \beta \right)^n
    = \alpha \sum_{i=0}^n \binom{n}{i} \left(u\sqrt{h} \beta \right)^i
    \le \alpha \sum_{i=0}^n \left(u\sqrt{h} \beta n \right)^i$$
    where in the last step we use the bound
    $\binom{n}{i} = \frac{n \cdot (n-1) \cdots (n-i+1)}{i!} \le \frac{n^i}{i!} \le n^i.$
    
    Plugging this to~\eqref{Eq:Calc02}, we obtain:
    \begin{align*}
        \int_{r_0}^\infty Q_n\left(\frac{r}{\sqrt{h}}\right) A(r) \, \dr 
        &\le \frac{\V(0) \cdot \alpha}{N_n} \int_{r_0/\sqrt{h}}^\infty  u^{n-1} \exp\left(-\frac{u^2}{2}\right) \, \sum_{i=0}^n \left(u\sqrt{h} \beta n \right)^i \, \du \\
        &= \frac{\V(0) \cdot \alpha}{N_n} \cdot \sum_{i=0}^n \int_{r_0/\sqrt{h}}^\infty  u^{n-1} \exp\left(-\frac{u^2}{2}\right) \,  \left(u\sqrt{h} \beta n \right)^i \, \du.
    \end{align*}
    For each $i \in \{0,1,\dots,n\}$, we have:
    \begin{align*}
        \int_{r_0/\sqrt{h}}^\infty  u^{n-1} \exp\left(-\frac{u^2}{2}\right) \,  \left(u\sqrt{h} \beta n \right)^i \, \du 
        &= \left(\sqrt{h} \beta n \right)^i \int_{r_0/\sqrt{h}}^\infty  u^{n+i-1} \exp\left(-\frac{u^2}{2}\right) \, \du \\
        &= N_{n+i} \left(\sqrt{h} \beta n \right)^i \int_{r_0/\sqrt{h}}^\infty  \gamma_{n+i}(u) \, \du \\
        &\le N_n \cdot \int_{r_0/\sqrt{h}}^\infty  \gamma_{n+i}(u) \, \du \\
        &= N_n \cdot  Q_{n+i}\left(\frac{r_0}{\sqrt{h}}\right).
    \end{align*}
    where the inequality above follows from our assumption $h \le 1/(2n^3 \beta^2)$ (see Lemma~\ref{Lem:GrowthConstant}).
    
    Combining the two parts above and plugging in the result back to~\eqref{Eq:Calc01} and~\eqref{Eq:Calc02}, we obtain:    \begin{align*}
        \pi_h\left(\X_{r_0}^\complement \right) 
        &\le \frac{1}{\V(0)} 
        \int_{r_0}^\infty Q_n\left(\frac{r}{\sqrt{h}}\right) A(r) \, \dr \\
        &\le \frac{\alpha}{N_n} \cdot \sum_{i=0}^n \int_{r_0/\sqrt{h}}^\infty  u^{n-1} \exp\left(-\frac{u^2}{2}\right) \,  \left(u\sqrt{h} \beta n \right)^i \, \du \\
        &\le \alpha \cdot \sum_{i=0}^n   Q_{n+i}\left(\frac{r_0}{\sqrt{h}}\right) \\
        &\le \alpha \cdot (n+1) \cdot  Q_{2n}\left(\frac{r_0}{\sqrt{h}}\right)
    \end{align*}
    where in the last inequality we use the trivial bound $Q_{n+i}(r) \le Q_{2n}(r)$ for $i \le n$ and $r \ge 0$.
\end{proof}

\subsubsection{Helper lemma on gradient of distance function}
\label{Sec:DistanceFunctionLipschitz}

\begin{lemma}\label{Lem:DistanceFunctionLipschitz}
    Let $\X \subset \R^n$ be a closed set with non-empty interior.
    Define $u \colon \R^n \to [0,\infty)$ by:
    $$u(y) := \dist(y,\X) = \min_{x \in \X} \|x-y\|.$$
    Then $u$ is differentiable almost everywhere, and if $u$ differentiable at $y \in \X^\complement$, then $\|\nabla u(y)\| = 1$.
\end{lemma}
\begin{proof}
    This follows the argument in~\cite[Proof of Lemma~5.4]{BNN13}.
    For any $x,y\in \R^n$, by triangle inequality we have:
    $$u(y) = \min_{z \in \X} \|y-z\|
    \le \min_{z \in \X} \left(\|y-x\| + \|x-z\| \right) = \|y-x\| + u(x).$$
    Exchanging $x$ and $y$ gives
    $|u(x) - u(y)| \le \|x-y\|$,
    which shows that $u$ is $1$-Lipschitz on $\R^n$.
    Then by Rademacher's theorem, $u$ is differentiable almost everywhere on $\R^n$.
    
    Now fix $y \in \X^\complement$ such that $u$ is differentiable at $y$. 
    Since $u$ is $1$-Lipschitz, we know $\|\nabla u(y)\| \le 1$.
    We will show $\|\nabla u(y)\| \ge 1$, which implies the claim $\|\nabla u(y)\| = 1$.    
    Since $\X$ is closed, there exists $x \in \X$ such that
    $u(y) = \|y-x\| > 0.$
    Define the unit vector
    $$v := \frac{y-x}{\|y-x\|}.$$
    Since $x \in \X$, for all $t \in (0, u(y))$ we have
    $$u(y-tv) \le \|y-tv-x\| = \left\| (y-x) \left(1-\frac{t}{\|y-x\|}\right) \right\| = \|y-x\| - t = u(y) - t.$$
    On the other hand, since $u$ is $1$-Lipschitz, we also have
    $$u(y-tv) \ge u(y) - \|tv\| = u(y) - t.$$
    Therefore, for all $t \in (0, u(y))$, we have $u(y-tv) = u(y) - t$.
    Hence, 
    $$\langle \nabla u(y), -v \rangle = \lim_{t \downarrow 0}\frac{u(y-tv)-u(y)}{t}
    = -1$$
    or equivalently, $\langle \nabla u(y), v \rangle = 1$.
    By Cauchy-Schwarz inequality, this implies
    $$1 = \langle \nabla u(y), v \rangle \le \|\nabla u(y) \| \cdot \|v\| = \|\nabla u(y)\|.$$
    This completes the proof.
\end{proof}

\subsubsection{Helper lemma on the normalizing constant}
\label{Sec:GrowthConstant}

Recall the normalizing constant $N_n = 2^{\frac{n}{2}-1} \Gamma(\frac{n}{2})$ for the chi distribution, where $\Gamma(m) = \int_0^\infty e^{-t} \, t^{m-1} \, dt$ is the Gamma function.
We have the following estimate.

\begin{lemma}\label{Lem:GrowthConstant}
    Assume $0 \le h \le 1/(2n^3 \beta^2)$ for some $n \ge 1$ and $\beta > 0$.
    Then for all $i \in \{0,1,\dots,n\}$,
    $$\left(\sqrt{h} n \beta \right)^i \cdot \frac{N_{n+i}}{N_n} \le 1.$$
\end{lemma}
\begin{proof}
    We recall a classical bound on Gamma function~\citep{wendel48} for any $x > 0$ and $0 < s < 1$:
    $$\Gamma(x+s) \le x^s \cdot \Gamma(x).$$
    Applying this bound with $s = \frac{1}{2}$, for each $i \in \{1,\dots,n\}$ we have:
    \begin{align*}
        \frac{\Gamma(\frac{n+i}{2})}{\Gamma(\frac{n}{2})}
        = \prod_{k=0}^{i-1} \frac{\Gamma(\frac{n+k}{2} + \frac{1}{2})}{\Gamma(\frac{n+k}{2})}
        \le \prod_{k=0}^{i-1} \left(\frac{n+k}{2}\right)^{1/2} \le \left(\frac{n+i-1}{2}\right)^{i/2} \le n^{i/2}.
    \end{align*}
    Therefore, for $i \in \{1,\dots,n\}$ we can bound:
    $$\frac{N_{n+i}}{N_n} = 2^{i/2} \cdot \frac{\Gamma(\frac{n+i}{2})}{\Gamma(\frac{n}{2})} \le (2n)^{i/2}.$$
    Therefore,
    $$(\sqrt{h} n \beta)^i \cdot \frac{N_{n+i}}{N_n} \le (\sqrt{2h} n^{3/2} \beta)^i \le 1$$
    since we assume 
    $h \le 1/(2n^3\beta^2)$.
\end{proof}

\subsection{Bound on expected failure probability under stationarity}
\label{App:SuccessStationary}

We recall the notations introduced in Section~\ref{Sec:SuccessStationary}.

\begin{replemma}{Lem:SuccProbStat}
    Let $\X \subset \R^n$ be a compact body which satisfies the $(\alpha,\beta)$-volume growth condition.
    Assume:
    \begin{align*}
        h &\le \frac{1}{2n^2 \beta^2 \cdot \max\{n, \, \log((n+1) \alpha S) \}} \\
        N &= 8\alpha S \log S
    \end{align*}
    where $S \ge 3$ is arbitrary.
    Then we have: 
    \begin{align*}
        F(h,N) \le \frac{3}{S}.
    \end{align*}
\end{replemma}
\begin{proof}
We will use Lemma~\ref{Lem:SuccProbStatGeneral} below.
Let 
$$r = \sqrt{8h \cdot \max\{n, \, \log((n+1)\alpha S)\}}.$$
Note this choice satisfies the constraint~\eqref{Eq:Constr1} for $r$.
The assumption on $h$ satisfies the constraint~\eqref{Eq:Constt1}.

Combining our choice of $r$ with the assumption on $h$:
$$r \le \sqrt{\frac{8 \cdot \max \{n, \log (\alpha (n+1)S)\}}{2n^2 \beta^2 \cdot \max\{n, \, \log(\alpha(n+1)S)\}}} = \frac{2}{n\beta}.$$
Therefore,
$(1+r\beta)^n \le \exp(rn\beta) \le e^2 < 8$.
Therefore, the choice $N = 8\alpha S \log S$ also satisfies the constraint~\eqref{Eq:ConstN1} for $N$.
With these choices, the bound from Lemma~\ref{Lem:SuccProbStatGeneral} yields:
\begin{align*}
    F(h,N) \le \frac{3}{S}.
\end{align*}
\end{proof}

\subsubsection{Helper lemma on a more general bound}

We can bound the expected failure probability under stationarity.
We specialize this in Lemma~\ref{Lem:SuccProbStat} by choosing specific values for $r$ and $N$.

\begin{lemma}\label{Lem:SuccProbStatGeneral}
    Let $\X \subset \R^n$ be a compact body which satisfies the $(\alpha,\beta)$-volume growth condition.
    Let $h,r > 0$ and $N \in \bN$ satisfy:
    \begin{align}
        h &\le \frac{1}{2n^3 \beta^2} \label{Eq:Constt1} \\
        r &\ge \sqrt{8h \cdot \max\{n, \, \log ((n+1) \alpha S)\}} \label{Eq:Constr1} \\
        N &\ge \alpha S \log S \cdot (1 + r\beta)^n \label{Eq:ConstN1}
    \end{align}
    where $S \ge 3$ is arbitrary.
    Then we have: 
    \begin{align*}
        F(h,N) \le \frac{3}{S}.
    \end{align*}
\end{lemma}
\begin{proof}
    Let $A_{h,N} \colon \R^n \to \R$ be the function:
    $$A_{h,N}(y) = (1-\ell_h(y))^N.$$
    We partition 
    $\R^n = \X_r^\complement \cup \B_1 \cup \B_2$
    where
    \begin{align*}
        \B_1 &:= \X_r \cap \left\{y \in \R^n \colon \ell_h(y) \ge \frac{1}{N} \log S \right\} \\
        \B_2 &:= \X_r \cap \left\{y \in \R^n \colon \ell_h(y) < \frac{1}{N} \log S \right\}.
    \end{align*}
    For simplicity, we suppress the dependence on the argument $y$ in the integrals below.
    We can split:
    \begin{align*}
        F(h,N) 
        &= \int_{\R^n} A_{h,N} \, \d\pi_h 
        = \int_{\X_r^\complement} A_{h,N} \, \d\pi_h 
        + \int_{\B_1} A_{h,N} \, \d\pi_h 
        + \int_{\B_2} A_{h,N} \, \d\pi_h.
    \end{align*}
    We control each part separately:
    
    \begin{enumerate}
        \item \textbf{Integral over $\X_r^\complement$:} 
        By the trivial bound 
        $A_{h,N}(y) = (1-\ell_h(y))^N \le 1$ for all $y \in \R^n$, and using the bound from  Lemma~\ref{Lem:EscProb} (which applies since we assume $h \le 1/(2n^3 \beta^2)$), we have:
        \begin{align*}
            \int_{\X_r^\complement} A_{h,N} \, \d\pi_h 
            \le \int_{\X_r^\complement} \d\pi_h = \pi_h(\X_r^\complement)
            \le \alpha (n+1) \cdot Q_{2n}\left(\frac{r}{\sqrt{h}}\right).
        \end{align*}
        Since $r \ge \sqrt{8hn}$, we have
        $r-\sqrt{2hn} \ge \frac{r}{2}$, so by the Gaussian tail bound~\eqref{Eq:TailBound} from Lemma~\ref{Lem:TailBound}:
        $$Q_{2n}\left(\frac{r}{\sqrt{h}}\right) \le \exp\left(-\frac{1}{2}\left(\frac{r}{\sqrt{h}}-\sqrt{2n}\right)^2\right) 
        = \exp\left(-\frac{\left(r-\sqrt{2hn}\right)^2}{2h}\right)
        \le \exp\left(-\frac{r^2}{8h}\right).$$
        Therefore, since we also assume $\frac{r^2}{8h} \ge \log((n+1)\alpha S)$, we can bound the integral above by:
        \begin{align*}
            \int_{\X_r^\complement} A_{h,N} \, \d\pi_h 
            \le \alpha (n+1) \cdot \exp\left(-\frac{r^2}{8h}\right) \le \frac{1}{S}.
        \end{align*}

        \item \textbf{Integral over $\B_1$:}
        By definition, for $y \in \B_1$ we have $\ell_h(y) \ge \frac{1}{N} \log S$, so 
        $$A_{h,N}(y) = (1 - \ell_h(y))^N \le \exp(-\ell_h(y) N) \le \frac{1}{S}.$$
        Therefore, since we also have $\pi_h(\B_1) \le \pi_h(\R^n) = 1$, we get
        $$\int_{\B_1} A_{h,N} \, d\pi_h \le \int_{\B_1} \frac{1}{S} \, d\pi_h = \frac{1}{S} \, \pi_h(\B_1) \le \frac{1}{S}.$$
    
        \item \textbf{Integral over $\B_2$:}
        We use the trivial bound $A_{h,N}(y) \le 1$ for all $y \in \R^n$, the formula~\eqref{Eq:FormulaDensity} for the density of $\pi_h$, and the bound $\ell_h(y) \le \frac{1}{N} \log S$ for $y \in \B_2$ by the definition of $\B_2$.
        We also use the inclusion $\B_2 \subseteq \X_r$, and the bound on $\Vol(\X_r)$ from the volume growth condition, to get:
        \begin{align*}
            \int_{\B_2} A_{h,N} \, \d\pi_h 
            &\le \int_{\B_2} 1 \cdot \d\pi_h(y) 
            = \int_{\B_2} \frac{\ell_h(y)}{\Vol(\X)}  \, \dy \\
            &\le \frac{\log S}{N} \cdot \int_{\B_2} \frac{1}{\Vol(\X)} \, \dy \\
            &\le \frac{\log S}{N} \cdot \int_{\X_r} \frac{1}{\Vol(\X)} \, \dy 
            = \frac{\log S}{N} \cdot \frac{\Vol(\X_r)}{\Vol(\X)} \\
            &\le \frac{\log S}{N} \cdot \alpha \cdot \left(1 + r \beta \right)^n \\
            &\le \frac{1}{S}
        \end{align*}
        where the last inequality follows from our assumption on $N$.
    \end{enumerate}
    
    Combining the three parts above, we get the desired bound:
    $F(h,N) 
    \le \frac{1}{S}  + \frac{1}{S} + \frac{1}{S} = \frac{3}{S}$.
\end{proof}

\subsubsection{Helper lemma on Gaussian tail bound}

Recall $Q_n(r) = \Pr(\|Z\|\ge r)$ where $Z \sim \N(0,I_n)$ in $\R^n$.

\begin{lemma}\label{Lem:TailBound}
For all $r\ge \sqrt{n}$:
\begin{align}\label{Eq:TailBound}
Q_n(r) \le \exp\left(-\frac{(r-\sqrt n)^2}{2}\right).
\end{align}
\end{lemma}
\begin{proof}
Note $x \mapsto \|x\|$ is $1$-Lipschitz, since by the triangle inequality, for any $x,y \in \R^n$:
$$\big|\|x\|-\|y\|\big|\le \|x-y\|.$$
Let $Z \sim \N(0, I_n)$ in $\R^n$.
By the Gaussian concentration inequality for Lipschitz functions~\cite[Theorem~5.6]{boucheron2013concentration},
for any $t \ge 0$ we have:
$$\Pr\left(\|Z\|-\E[\|Z\|] \ge t \right) \le \exp\left(-\frac{t^2}{2}\right).$$
Let $r\ge \sqrt{n}$ and set $t := r-\E[\|Z\|]\ge 0$. Then
$$Q_n(r) = \Pr(\|Z\|\ge r) 
= \Pr\big(\|Z\|-\E[\|Z\|]\ge r-\E[\|Z\|]\big)
\le \exp\left(-\frac{(r-\E[\|Z\|])^2}{2}\right).$$
By Cauchy-Schwarz inequality, $\E[\|Z\|] \le \sqrt{\E[\|Z\|^2]}=\sqrt{n}$,
so $r-\E[\|Z\|] \ge r-\sqrt{n}$.
Since $a\mapsto e^{-a^2/2}$ is decreasing for $a\ge 0$, we have
$\exp\big(-\frac{(r-\E[\|Z\|])^2}{2}\big) \le \exp\big(-\frac{(r-\sqrt{n})^2}{2}\big),$
which proves the claim.
\end{proof}

\subsection{Bound on the expected number of trials under stationarity}
\label{App:NumTrialsStationary}

We recall the notations introduced in Section~\ref{Sec:NumTrialsStationary}.

\begin{replemma}{Lem:NumTrialsStatCor}
    Let $\X \subset \R^n$ be a compact body that satisfies the $(\alpha,\beta)$-volume growth condition for some $\alpha \in [1,\infty)$ and $\beta \in (0,\infty)$.
    Assume:
    \begin{align*}
        h &\le \frac{1}{2n^2 \beta^2 \cdot \max\{n, \, \log((n+1)\alpha S)\}} \\
        N &= 8\alpha S \log S
    \end{align*}
    where $S \ge 3$ is arbitrary.
    Then:
    $$\E_{\pi_h}[\mu_{h,N}] \le 16 \alpha \log S.$$
\end{replemma}
\begin{proof}
We will use Lemma~\ref{Lem:NumTrialsStat} below.
Let 
$$r = \sqrt{8h \cdot \max\{n, \, \log((n+1)\alpha S)\}}.$$
Note this choice satisfies the constraint~\eqref{Eq:Constr} for $r$.
The assumption on $h$ satisfies the constraint~\eqref{Eq:Constt}.

Combining our choice of $r$ with the assumption on $h$:
$$r \le \sqrt{\frac{8 \cdot \max \{n, \log (\alpha (n+1)S)\}}{2n^2 \beta^2 \cdot \max\{n, \, \log(\alpha(n+1)S)\}}} = \frac{2}{n\beta}.$$
Therefore,
$(1+r\beta)^n \le \exp(rn\beta) \le e^2 < 8$.
Therefore, the choice $N = 8\alpha S \log S$ also satisfies the constraint~\eqref{Eq:ConstN} for $N$.
With these choices, the bound from Lemma~\ref{Lem:NumTrialsStat} yields:
\begin{align*}
    \E_{\pi_h}[\mu_{h,N}] 
    \le \frac{2N}{S}
    = 16 \alpha \log S.
\end{align*}
\end{proof}

\subsubsection{Helper lemma on a more general bound}

We can bound the expected number of trials under stationarity.
We specialize this in Lemma~\ref{Lem:NumTrialsStatCor} by choosing specific values for $r$ and $N$.

\begin{lemma}\label{Lem:NumTrialsStat}
    Let $\X \subset \R^n$ be a compact body that satisfies the $(\alpha,\beta)$-volume growth condition for some $\alpha \in [1,\infty)$ and $\beta \in (0,\infty)$.
    Assume $h,r > 0$ and $N \in \bN$ satisfy:
    \begin{align}
        h &\le \frac{1}{2n^3 \beta^2} \label{Eq:Constt}  \\
        N &\ge \alpha S \log S \cdot (1+r\beta)^n \label{Eq:ConstN} \\
        r &\ge \sqrt{8h \cdot \max\{n, \, \log ((n+1)\alpha S)\}} \label{Eq:Constr} 
    \end{align}
    where $S \ge 3$ is arbitrary.
    Then:
    $$\E_{\pi_h}[\mu_{h,N}] \le \frac{2N}{S}.$$
\end{lemma}
\begin{proof}
We split the integral into two parts:
\begin{align*}
\E_{\pi_h}[\mu_{h,N}] 
&= \int_{\R^n} \mu_{h,N} \, \d\pi_h 
= \int_{\X_r} \mu_{h,N} \, \d\pi_h + \int_{\X_r^\complement} \mu_{h,N} \, \d\pi_h.
\end{align*}
We can bound each integral above as follows:
\begin{enumerate}
    \item \textbf{Integral over $\X_r$:}
    Using the bound $\mu_{h,N}(y) \le \frac{1}{\ell_h(y)}$ from~\eqref{Eq:PtwiseBound}, the formula~\eqref{Eq:FormulaDensity} for the density of $\pi_h$, and the volume growth bound, we get:
    \begin{align*}
        \int_{\X_r} \mu_{h,N} \, d\pi_h 
        &\le \int_{\X_r} \frac{1}{\ell_h(y)} \cdot \pi_h(y) \, dy 
        = \int_{\X_r} \frac{1}{\Vol(\X)} \, dy 
        = \frac{\Vol(\X_r)}{\Vol(\X)}
        \le \alpha \cdot (1+r\beta)^n 
        \le \frac{N}{S \log S}.
    \end{align*}
    The last inequality follows from our assumption on $N$.
    Since we assume $S \ge 3$, $\log S \ge 1$, so we can further bound $\frac{N}{S \log S} \le \frac{N}{S}$.
    \item \textbf{Integral over $\X_r^\complement$:}
    Using the bound $\mu_{h,N}(y) \le N$ from~\eqref{Eq:PtwiseBound}, and the bound from Lemma~\ref{Lem:EscProb} (which applies since we assume $h \le 1/(2n^3 \beta^2)$), we can bound:
    \begin{align*}
        \int_{\X_r^\complement} \mu_{h,N} \, d\pi_h 
        &\le 
        N \cdot \pi_h(\X_r^\complement) 
        \le  N \cdot \alpha (n+1) \cdot Q_{2n}\left(\frac{r}{\sqrt{h}}\right)
        \le \frac{N}{S}.
    \end{align*}
    The last inequality follows since we assume $r \ge \sqrt{8hn}$, so
    $\frac{r}{\sqrt{h}}-\sqrt{2n} \ge \frac{r}{2\sqrt{h}}$; and by the Gaussian tail bound (see Lemma~\ref{Lem:TailBound}):
    $Q_{2n}\left(\frac{r}{\sqrt{h}}\right) \le \exp\left(-\frac{1}{2}\left(\frac{r}{\sqrt{h}}-\sqrt{2n}\right)^2\right) 
    \le \exp\left(-\frac{r^2}{8h}\right)$; and since we assume $r \ge \sqrt{8h \cdot \log ((n+1)\alpha S)}$, so that $\exp\left(-\frac{r^2}{8h}\right) \le \frac{1}{(n+1)\alpha S}$.
\end{enumerate}

Combining the two estimates above, we have the bound:
\begin{align*}
    \E_{\pi_h}[\mu_{h,N}] 
    \le \frac{N}{S} + \frac{N}{S} = \frac{2N}{S}.
\end{align*}
\end{proof}

\subsection{Proof of Lemma~\ref{Lem:SuccProbInAndOut}}
\label{Sec:SuccProbInAndOutProof}

\begin{replemma}{Lem:SuccProbInAndOut}
    Let $\pi \propto \one_\X$ where $\X \subset \R^n$ satisfies the $(\alpha,\beta)$-volume growth condition for some $\alpha \in [1,\infty)$ and $\beta \in (0,\infty)$. 
    Let $x_0 \sim \rho_0$ where $\rho_0$ is $M$-warm with respect to $\pi$ for some $M \in [1,\infty)$.
    Assume:
    \begin{align*}
        h &\le \frac{1}{2\beta^2 n^3 \cdot \max\{1, \, \frac{1}{n} \log((n+1)\alpha S)\}} \\
        N &= 8\alpha S \log S
    \end{align*}
    where $S \ge 3$ is arbitrary.
    Assume \textbf{In-and-Out} succeeds for $i \ge 0$ iterations (i.e., conditioned on the event ``reaches iteration $i$''), and is currently at $x_i \sim \rho_i$ which is $M$-warm with respect to $\pi$.
    Then the following hold:
    \begin{enumerate}
        \item The next iteration (to compute $x_{i+1}$) succeeds with probability:
        $$\Pr_{x_i \sim \rho_i} \left(\text{next iteration succeeds} \mid \text{reach iteration $i$}\right) \ge 1-\frac{3M}{S}.$$
        \item The expected number of trials in the next iteration is
        $$\E_{x_i \sim \rho_i} \left[\text{$\#$ trials until success in next iteration} \mid \text{reach iteration $i$} \right] \le 16 M\alpha  \log S.$$
        \item Upon accepting $x_{i+1} \sim \rho_{i+1}$, the next distribution $\rho_{i+1}$ is $M$-warm with respect to $\pi$.
    \end{enumerate}
\end{replemma}
\begin{proof}
Suppose we are at $x_i \sim \rho_i$ in the $i$-th iteration.
When $i = 0$, we know $\rho_0$ is $M$-warm by assumption.
For $i > 0$, we assume $\rho_i$ is $M$-warm as an inductive hypothesis.

In the forward step of the algorithm, we draw $y_{i} \sim \N(x_i, hI_n)$ to obtain $y_{i}$ with marginal distribution $\rho_{i}^Y := \rho_i \ast \N(0, hI_n)$.
Note $\rho_{i}^Y$ is $M$-warm with respect to $\pi_h = \pi \ast \N(0,hI_n)$.
Indeed, letting $\N_h(y)$ denote the density of $\N(0,hI_n)$ at $y \in \R^n$, we have by definition:
$$\rho_{i}^Y(y) = \int_{\X} \rho_i(x) \, \N_h(y-x) \, \dx \le M \cdot \int_{\X} \pi(x) \, \N_h(y-x) \, \dx = M \cdot \pi_h(y).$$

In the backward step, we try to draw $x_{i+1}' \sim \N(y_{i}, hI_n)$ and accept when $x_{i+1}' \in \X$ (in which case we set $x_{i+1} = x_{i+1}'$), and repeat up to $N$ times.
We analyze this as follows:

\begin{enumerate}
    \item \textbf{Failure probability:}
    Conditioning on $y_{i}$, the failure probability of one trial in the next iteration is
    $$\Pr(x_{i+1}' \notin \X \mid y_{i}) = \Pr(y_{i} + \sqrt{h} Z \notin \X \mid y_{i}) = 1-\ell_h(y_{i})$$
    where $Z \sim \N(0, I_n)$ in $\R^n$, and recall $\ell_h$ is the local conductance defined in~\eqref{Eq:LocalConductance}.
    If the trial fails, then we repeat it (conditioning on the same $y_{i}$) for $\le N$ trials.
    Then the failure probability over $N$ trials (conditioned on $y_{i}$) is $( 1-\ell_h(y_{i}))^N$.
    Taking expectation over $y_{i} \sim \rho_{i}^Y$, the failure probability over $N$ trials is 
    $\E_{\rho_{i}^Y}[(1-\ell_h)^N]$; this is the probability that the next iteration fails.
    Since $\rho_{i}^Y$ is $M$-warm with respect to $\pi_h$, we can bound this by the expectation under stationary:
    \begin{align*}
        \E_{\rho_{i}^Y}[(1-\ell_h)^N]
        &\le M \cdot \E_{\pi_h}\left[(1-\ell_h)^N \right] 
        \le \frac{3M}{S}
    \end{align*}
    where the last inequality follows from Lemma~\ref{Lem:SuccProbStat} (which holds for our choices of $h$ and $N$).
    \item \textbf{Controlling the expected number of trials:}
    Conditioning on $y_{i}$,
    let $M_{h,N}(y_{i})$ denote the number of trials in the backward step, where each trial is independent and has success probability $\ell_h(y_{i})$, and we run at most $N$ trials.
    Then $M_{h,N}(y_{i}) \stackrel{d}{=} \min\{G_h(y_{i}), N\}$ where $G_h(y_{i})$ is a geometric random variable with success probability $\ell_h(y_{i})$.
    The expected number of trials (conditioned on $y_i$, so the expectation is over the geometric random variable $G_h(y_{i})$) is
    $\mu_{h,N}(y_{i}) = \E[\min\{G_h(y_{i}), N\}]$, which we also defined in~\eqref{Eq:NumTrialsStat}.
    Taking expectation over $y_{i} \sim \rho_{i}^Y$, the expected number of trials in the backward step is $\E_{\rho_{i}^Y}[\mu_{h,N}]$.
    Since $\rho_{i}^Y$ is $M$-warm with respect to $\pi_h$, we can bound this by:
    \begin{align*}
        \E_{\rho_{i}^Y}[\mu_{h,N}]
        &\le M \cdot \E_{\pi_h}[\mu_{h,N}(y)] 
        \le 16M \alpha \log S
    \end{align*}
    where the last inequality follows from Lemma~\ref{Lem:NumTrialsStatCor} (which holds for our choices of $h$ and $N$).

    \item \textbf{Warmness:} 
    Upon acceptance, the next random variable $x_{i+1} \sim \rho_{i+1}$ has distribution:
    $$\rho_{i+1}(x) = \int_{\R^{n}} \rho_i^Y(y) \cdot \pi^{X \mid Y}(x \mid y) \, \dy$$
    where $\pi^{X \mid Y}(\cdot\mid y) \propto \N(y, h I_{n}) \cdot \one_\X$.
    Since $\rho_i^Y$ is $M$-warm with respect to $\pi_h = \pi^Y$, we get:
    $$\rho_{i+1}(x) \le M \cdot \int_{\R^{n}} \pi^Y(y) \cdot \pi^{X \mid Y}(x \mid y) \, \dy = M \cdot \pi^X(x)$$
    which shows $\rho_{i+1}$ is $M$-warm with respect to $\pi^X$.
\end{enumerate}
\end{proof}

\subsection{Proof of Theorem~\ref{Thm:Nonconvex}}
\label{App:NonconvexProof}

We are now ready to prove Theorem~\ref{Thm:Nonconvex}.

\begin{proof}[Proof of Theorem~\ref{Thm:Nonconvex}]
    Given the target error $\varepsilon \in (0,\frac{1}{2})$, we define $\varepsilon' := \frac{1}{2} \varepsilon$ and $\eta := \frac{1}{8} \varepsilon$, so $\varepsilon' < \frac{1}{4}$ and $\eta < \frac{1}{16}$.
    By assumption, $x_0 \sim \rho_0$ is $M$-warm with respect to $\pi \propto \one_\X$.
    In particular, $\sR_q(\rho_0 \,\|\, \pi) \le \log M$.
    We choose the number of iterations $T$ to be:
    \begin{align}
        T &:= 8q \CPI \beta^2 n^2 \left(n + \log \frac{3(n+1) \alpha M}{\eta} \right) \cdot \log \left(\frac{M}{\varepsilon'}\right) \times \cdots \notag \\
        &\qquad \log \left( 4q \CPI \beta^2 n^2 \left(n + \log \frac{3(n+1) \alpha M}{\eta} \right) \cdot \log \left(\frac{M}{\varepsilon'}\right) \right). \label{Eq:TPIDef}
    \end{align}
    We define an auxiliary parameter $S$ to be:
    \begin{align}\label{Eq:Sdefeq}
        S :=  \frac{3TM}{\eta}.
    \end{align}
    We also define the step size $h$ and the threshold on the number of trials $N$ in each iteration to be:
    \begin{align}
        h &:= \frac{1}{2\beta^2 \, n^3 \, (1 + \frac{1}{n}\log((n+1) \alpha S))}  
        \label{Eq:hDef} \\
        N &:= 8\alpha S \log S. \label{Eq:NDef}
    \end{align}
    We note $S \ge 3$, $h \le \frac{1}{2n} \le \frac{1}{2}$, and our choices of $h$ and $N$ above satisfy the assumptions in Lemma~\ref{Lem:SuccProbInAndOut}.

    By Lemma~\ref{Lem:SuccProbInAndOut} inductively for $i \ge 0$, conditioned on the event that \textbf{In-and-Out} reaches iteration $i$, the iterate $x_i \sim \rho_i$ of \textbf{In-and-Out} remains $M$-warm.
    Then we can do the following analysis.

    \paragraph{(1) Failure probability:}
    For each $i \ge 0$, let $E_i$ denote the event that \textbf{In-and-Out} reaches iteration $i$ (i.e., the first $i$ iterations succeed) starting from $x_0 \sim \rho_0$.
    Then $\Pr(E_0) = 1$, and we want to bound $\Pr(E_T)$.
    By Lemma~\ref{Lem:SuccProbInAndOut} and the warmness guarantee, for each $i \in \{0,1,\dots,T-1\}$ we have
    \begin{align*}
        \Pr\left(E_{i+1}^\complement \mid E_i \right)
        \le \frac{3M}{S} \stackrel{\eqref{Eq:Sdefeq}}{=} \frac{\eta}{T}.
    \end{align*}
    Since the events $E_i$ are nested ($E_{i+1} \subseteq E_i$ for $i \ge 0$), we can decompose:
    \begin{align*}
        \Pr\left(E_T^\complement \right)
        &= \sum_{i=0}^{T-1} \Pr\left(E_i \cap E_{i+1}^\complement \right)
        = \sum_{i=0}^{T-1} \Pr\left(E_i \right) \cdot \Pr\left(E_{i+1}^\complement \mid E_i \right) 
        \le \sum_{i=0}^{T-1} 1 \cdot \frac{\eta}{T} 
        = \eta.
    \end{align*}
    Therefore, the probability that \textbf{In-and-Out} runs successfully for $T$ iterations is:
    $$\Pr\left(E_T \right) = 1 - \Pr\left(E_T^\complement \right) \ge 1-\eta = 1 - \frac{\varepsilon}{8}.$$

    \paragraph{(2) Error guarantee:}
    In view of Lemma~\ref{Lem:OuterConvergence},
    we define the initial duration when R\'enyi divergence decreases slowly along \textbf{In-and-Out}:
    \begin{align}\label{Eq:TPIDef0}
        T_0 
        &:= \max\left\{ 0, \; \left\lceil \frac{q \left(\sR_q(\rho_0 \,\|\, \pi)-1\right)}{\log\left(1 + \frac{h}{\CPI}\right)} \right\rceil \right\}.
    \end{align}
    After $T_0$ iterations, the R\'enyi divergence decreases exponentially fast.
    We define:
    \begin{align}\label{Eq:TPIDef2}
        \tilde T &:= T_0 + \frac{ q \log \frac{1}{\varepsilon'}}{\log\left(1 + \frac{h}{\CPI} \right)}.
    \end{align}
    We claim that our choice of $T$ in~\eqref{Eq:TPIDef} satisfies $T \ge \tilde T$; we show this below.
    Assuming this claim, if \textbf{In-and-Out} runs successfully for $T \ge \tilde T$ iterations (which holds with probability at least $1-\eta$ by part (1)), then by Lemma~\ref{Lem:OuterConvergence}, the output $x_T \sim \rho_T$ satisfies: 
    $$\sR_q(\rho_T \,\|\, \pi) 
    \,\le\, \left(1 + \frac{h}{\CPI}\right)^{-\frac{1}{q} (T-T_0)} + 4\eta
    \,\le\, \left(1 + \frac{h}{\CPI}\right)^{-\frac{1}{q} (\tilde T-T_0)} + 4\eta
    \,\stackrel{\eqref{Eq:TPIDef2}}{\le}\, \varepsilon' + 4\eta \,=\, \varepsilon.$$
    Thus, if \textbf{In-and-Out} succeeds for $T$ iterations, then its output satisfies the desired error guarantee.

    We now show $T \ge \tilde T$.
    Since $h \le \frac{1}{2}$ and $\CPI \ge 1$, we have $\frac{h}{\CPI} \le \frac{1}{2}$.
    Using the inequality $\log(1+t) \ge \frac{t}{2}$ for $0 < t = \frac{h}{\CPI} \le \frac{1}{2}$, we have $\log(1 + \frac{h}{\CPI}) \ge \frac{h}{2\CPI}$.
    Using $\sR_q(\rho_0 \,\|\, \pi) \le \log M$, we can estimate~\eqref{Eq:TPIDef0} by:
    $T_0 \le \frac{2q \CPI}{h}  \log M.$
    Furthermore, we can estimate~\eqref{Eq:TPIDef2} by:
    $$\tilde T \le \frac{2q \CPI}{h}  \log M + \frac{2q \CPI}{h} \log \frac{1}{\varepsilon'} \le \frac{2q \CPI}{h}  \log \frac{M}{\varepsilon'}.$$
    It remains to show that with our choices of $T$ in~\eqref{Eq:TPIDef} and $h$ in~\eqref{Eq:hDef}, we have $T \ge \frac{2q \CPI}{h}  \log \frac{M}{\varepsilon'}$.
    This follows by a routine computation, which we provide in Lemma~\ref{Lem:TBound}.

    \paragraph{(3) Expected number of trials:}
    For each $i \ge 0$, let $N_i$ denote the number of trials in iteration $i$ of \textbf{In-and-Out} to compute $x_{i+1} \sim \rho_{i+1}$.
    Recall from part (1) the event $E_i$ that \textbf{In-and-Out} reaches iteration $i$.
    If \textbf{In-and-Out} fails before reaching iteration $i$ (i.e.\ on the event $E_i^{\complement}$), then $N_i = 0$.
    By Lemma~\ref{Lem:SuccProbInAndOut}, if \textbf{In-and-Out} reaches iteration $i$ (i.e.\ on the event $E_i$), then the expected number of trials in the next iteration is:
    \begin{align}\label{Eq:PerIterPIEst}
        \E\left[N_i \mid E_i \right] 
        \,\le\, 16 M\alpha \log S 
        \,\stackrel{\eqref{Eq:Sdefeq}}{=}\, 16 M \alpha \log \frac{3MT}{\eta}.
    \end{align}
    Therefore, the expected number of trials over $T$ iterations of \textbf{In-and-Out} is:
    \begin{align}
        \E\big[\# & \text{ of trials in $T$ iterations of \textbf{In-and-Out}} \big] \notag \\
        &= \sum_{i = 0}^{T-1} \E\left[N_i \right] \notag \\
        &= \sum_{i = 0}^{T-1} \left(\Pr\left(E_i^\complement\right) \times \E\left[N_i \mid E_i^\complement \right] + \Pr\left(E_i \right) \times \E\left[N_i \mid E_i \right] \right) \notag \\
        &= \sum_{i = 0}^{T-1} \Pr\left(E_i\right) \times \E\left[N_i \mid E_i \right] \notag \\
        &\stackrel{\eqref{Eq:PerIterPIEst}}{\le} 16 M\alpha T \log \frac{3MT}{\eta}.
    \end{align}    
    
    To simplify this, we plug in our choice of $T$ from~\eqref{Eq:TPIDef} which is of the form $T = 2z \log z$, where
    \begin{align}\label{Eq:TdefXHelper}
        z := 4q \CPI \beta^2 n^2 \left(n + \log \frac{3(n+1) \alpha M}{\eta} \right) \cdot \log \left(\frac{M}{\varepsilon'}\right).
    \end{align}
    Using $\log z \le z-1 \le z$, we have
    \begin{align}\label{Eq:NumTrialsBd0}
        \log \frac{3MT}{\eta} = \log \frac{6Mz}{\eta} + \log (\log z) \le \log \frac{6Mz}{\eta} + \log z \le 2\log \frac{6Mz}{\eta}.
    \end{align}
    Therefore, the expected number of trials over $T$ iterations of \textbf{In-and-Out} is bounded by:
    \begin{align}
        \E\big[\# \text{ of trials in $T$ iterations of \textbf{In-and-Out}} \big]
        &\stackrel{\eqref{Eq:NumTrialsBd0}}{\le} 64 M \alpha z (\log z) \cdot \log \frac{6Mz}{\eta} \notag \\
        &\le 64 M \alpha z \cdot \left( \log \frac{6Mz}{\eta} \right)^2. \label{Eq:NumTrialsBd1}
    \end{align}
    The expression~\eqref{Eq:NumTrialsBd1} is a precise bound on the expected number of total trials of \textbf{In-and-Out}, where $z$ is defined in~\eqref{Eq:TdefXHelper}.
    To simplify this further, we use the $\tilde O$ notation to hide constants and logarithmic dependencies on the parameters, so for example, $M \log M = \tilde O(M)$ and $z \log z = \tilde O(z)$.
    Recalling $\varepsilon' = \frac{1}{2} \varepsilon$ and $\eta = \frac{1}{8} \varepsilon$, we can simplify the expression in~\eqref{Eq:NumTrialsBd1} as:
    \begin{align*}
        64 M \alpha z \cdot \left( \log \frac{6Mz}{\eta} \right)^2
        &= \tilde O\left(M \alpha z \left( \log \frac{1}{\eta} \right)^2 \right) \\
        &\stackrel{\eqref{Eq:TdefXHelper}}{=} \tilde O\left(M \alpha \cdot q \CPI \beta^2 n^3 \left(1 + \frac{1}{n} \log \frac{3(n+1) \alpha M}{\eta} \right) \cdot \left(\log \frac{M}{\varepsilon'}\right) \cdot \left( \log \frac{1}{\eta} \right)^2 \right) \\
        &= \tilde O\left(q \CPI \, \alpha \beta^2 \, M \, n^3 \left(1 + \frac{1}{n} \log \frac{1}{\eta} \right) \cdot \left(\log \frac{1}{\varepsilon'} \right) \cdot \left( \log \frac{1}{\eta} \right)^2 \right) \\
        &= \tilde O\left(q \CPI \, \alpha \beta^2 \, M \, n^3 \cdot \left( \log \frac{1}{\varepsilon} \right)^4 \right)
    \end{align*}
    as claimed in the theorem. 

    Finally, inspecting the algorithm, we see that each trial requires one call to the membership oracle $\one_\X$ and $O(n)$ arithmetic operations.     
\end{proof}

\subsubsection{Helper lemma on bound on number of iterations}
\label{Sec:TBound}

\begin{lemma}\label{Lem:TBound}
    Under the same assumptions as in Theorem~\ref{Thm:Nonconvex}, and with the definitions of $T$ in~\eqref{Eq:TPIDef}, $h$ in~\eqref{Eq:hDef}, and $S$ in~\eqref{Eq:Sdefeq}, we have
    \begin{align}\label{Eq:TBoundClaim}
        T \ge \frac{2q \CPI}{h}  \log \frac{M}{\varepsilon'}.
    \end{align}
\end{lemma}
\begin{proof}
    Our choice~\eqref{Eq:TPIDef} of $T$ is of the form $T = 2z\log z$, where 
    $$z := 4q \CPI \beta^2 n^2 \left(n + \log \frac{3(n+1) \alpha M}{\eta} \right) \cdot \log \left(\frac{M}{\varepsilon'}\right) 
    \ge 32 \beta^2 n^2 \ge 32$$ 
    since $q \ge 2$, $\CPI \ge 1$, $n \ge 2$, $\log \frac{3(n+1) \alpha M}{\eta} \ge 2$, $\log \left(\frac{M}{\varepsilon'}\right) \ge 1$, and we assume $\beta \ge \frac{1}{n}$.
    We recall (see Lemma~\ref{Lem:LogBound}) that for $z \ge 2$, $y \ge 2z \log z$ implies $y/ \log y \ge z$.
    Then we observe that our choice of $T$ in~\eqref{Eq:TPIDef} implies:
    \begin{align}
        \frac{T}{\log T} 
        &\ge 4q \CPI \beta^2 n^2 \left(n + \log \frac{3(n+1) \alpha M}{\eta} \right) \cdot \log \frac{M}{\varepsilon'} \notag \\
        &\ge 16q \CPI \beta^2 n^2 \cdot \log \frac{M}{\varepsilon'} \label{Eq:TBoundObserved}
    \end{align}
    where the last inequality holds since $n \ge 2$ and $\log \frac{3(n+1) \alpha M}{\eta} \ge \log 9 > 2$.
    
    The right-hand side of the claim~\eqref{Eq:TBoundClaim} is, using the definitions of $h$ from~\eqref{Eq:hDef} and $S$ from~\eqref{Eq:Sdefeq}:
    \begin{align*}
        \mathsf{RHS} &:= \frac{2q \CPI}{h}  \log \frac{M}{\varepsilon'} \\
        &\stackrel{\eqref{Eq:hDef}}{=} 4q \beta^2 \CPI \, n^2 \left(n + \log ((n+1) \alpha S)\right) \cdot \log \frac{M}{\varepsilon'} \\
        &\stackrel{\eqref{Eq:Sdefeq}}{=} 4q \beta^2 \CPI \, n^2 \left(n + \log \left(\frac{3(n+1) \alpha M}{\eta}\right) + \log T \right) \cdot \log \frac{M}{\varepsilon'}.
    \end{align*}
    On the other hand, from our choice of $T = 2z \log z$ in~\eqref{Eq:TPIDef} with $\log z \ge 1$, we have:
    \begin{align*}
        T 
        &\ge \frac{T}{2} + z \\
        &= \frac{T}{2} + 4q \CPI \beta^2 n^2 \left(n + \log \frac{3(n+1) \alpha M}{\eta} \right) \cdot \log \left(\frac{M}{\varepsilon'}\right) \\
        &= \frac{T}{2} + \mathsf{RHS} - 4q \beta^2 \CPI \, n^2 \cdot \log \left(\frac{M}{\varepsilon'}\right) \cdot \log T.
    \end{align*}
    Therefore, to show $T \ge \mathsf{RHS}$, it suffices to show that:
    \begin{align*}
        \frac{T}{\log T} \ge 8 q \CPI \beta^2 n^2 \cdot \log \left(\frac{M}{\varepsilon'}\right)
    \end{align*}
    which we observed in~\eqref{Eq:TBoundObserved} for our choice of $T$.
    This completes the proof.
\end{proof}

\subsubsection{Helper lemma on logarithm}
\label{Sec:LogBound}

\begin{lemma}\label{Lem:LogBound}
    If $z \ge 2$ and $y \ge 2 z \log z$, then $y / \log y \ge z$.
\end{lemma}
\begin{proof}
    Define the function $\phi \colon (0,\infty) \to \R$ by $h(y) = \frac{y}{\log y}$.
    Its derivative is $\phi'(y) = \frac{\log y - 1}{(\log y)^2}$, so $\phi(y)$ is increasing for $y \ge e$.
    Since $y \ge 2z \log z \ge 4 \log 2 > e$, we have 
    $$\frac{y}{\log y} \ge \frac{2z \log z}{\log(2z \log z)}.$$
    Then to show $y / \log y \ge z$, it suffices to show $2 \log z \ge \log(2z \log z)$.
    Exponentiating and simplifying, this is equivalent to showing $z \ge 2 \log z$, which is true for $z \ge 2$.
\end{proof}

\section{Details for the volume growth condition}

\subsection{Volume growth condition for convex bodies}
\label{Sec:ReviewConvex}

\begin{replemma}{Lem:VolGrowthConvex}
If $\X \subset \R^n$ is convex, then it satisfies the $(1,\frac{1}{n} \xi(\X))$-volume growth condition.
\end{replemma}
\begin{proof}
    Let $\phi(t) = \log \Vol(\X_t)$ so $\phi'(t) = \xi(\X_t)$.
    By the Brunn-Minkowski theorem, we know $t \mapsto \Vol(\X_t)^{1/n} = \exp\left(\frac{1}{n} \phi(t)\right)$ is a concave function.
    This means
    \begin{align*}
        0 \ge \frac{\d^2}{\dt^2} \Vol(\X_t)^{1/n}
        &= \frac{\d^2}{\dt^2} \exp\left(\frac{1}{n} \phi(t)\right) \\
        &= \frac{\d}{\dt} \left(\frac{1}{n} \phi'(t) \exp\left(\frac{1}{n} \phi(t)\right) \right) \\
        &= \frac{1}{n^2} \left(n\phi''(t)  + (\phi'(t))^2 \right) \exp\left(\frac{1}{n} \phi(t)\right).
    \end{align*}
    Therefore, $n\phi''(t)  + (\phi'(t))^2 \le 0$.
    Equivalently, for all $t > 0$:
    $-\frac{\phi''(t)}{(\phi'(t))^2} \ge \frac{1}{n}$.
    We can write this as:
    $\frac{\d}{\dt} \frac{1}{\phi'(t)} = -\frac{\phi''(t)}{(\phi'(t))^2} \ge \frac{1}{n}$.
    Therefore,
    $\frac{1}{\phi'(t)} - \frac{1}{\phi'(0)} \ge \frac{t}{n}$.
    Equivalently,
    \begin{align*}
        \xi(\X_t) = \phi'(t) \le \frac{1}{\frac{1}{\xi(\X)} + \frac{t}{n}} = \frac{n}{t + \frac{n}{\xi(\X)}}
        = n \cdot \frac{\d}{\dt} \log\left(t + \frac{n}{\xi(\X)}\right).
    \end{align*}
    Therefore, $\int_0^t \xi(\X_s) \, \ds \le \int_0^t \frac{n}{s + \frac{n}{\xi(\X)}} \ds 
        = n \log \left( \frac{t + \frac{n}{\xi(\X)}}{\frac{n}{\xi(\X)}} \right)
        = n \log \left(1 + t \cdot \frac{\xi(\X)}{n} \right)$.
    This shows the claim:
    \begin{align*}
        \frac{\Vol(\X_t)}{\Vol(\X)} = \exp\left( \int_0^t \xi(\X_s) \, \ds \right) \le \left(1 + t \cdot \frac{\xi(\X)}{n} \right)^n.
    \end{align*}
\end{proof}

\subsection{Volume growth condition for star-shaped bodies}
\label{Sec:VolGrowthStarShaped}

\begin{replemma}{Lem:VolGrowthStarShaped}
    Let $\X \subset \R^n$ be a star-shaped body, so $\X = \bigcup_{i \in \I} \X^i$ where $\X^i$ is a convex body for each $i \in \I$ in a finite index set $\I$, and they share a common intersection $\Y = \X^i \cap \X^j \neq \emptyset$ for all $i \neq j$.
    Assume $\Y$ contains a ball of radius $r > 0$ centered at $0$, i.e., $B_r \subseteq \Y$.
    Then $\X$ satisfies the $(1, \frac{1}{r})$-volume growth condition.
\end{replemma}
\begin{proof}
    Fix $t > 0$.
    Since $\X = \bigcup_{i \in \I} \X^i$, we observe that $\X_t \subseteq \bigcup_{i \in \I} (\X^i)_t$;
    this is because any $x \in \X_t = \X \oplus B_t$ is of the form $x = y + z$ where $y \in \X^i$ for some $i \in \I$ and $z \in B_t$, so $x \in \X^i \oplus B_t = (\X^i)_t$.
    Next, since $B_r \subseteq \Y \subseteq \X^i$, we have $B_t = \frac{t}{r} B_r \subseteq \frac{t}{r} \X^i$.
    Then we have $(\X^i)_t = \X^i \oplus B_t \subseteq \X^i \oplus \frac{t}{r} \X^i = \left(1 + \frac{t}{r} \right) \X^i$, where the last equality holds since $\X^i$ is convex and contains $0$.
    Furthermore, since $\X^i \subseteq \X$, we also have $\left(1 + \frac{t}{r} \right) \X^i \subseteq \left(1 + \frac{t}{r} \right) \X$.
    Combining the relations above, we obtain
    $$\X_t 
    \,\subseteq\, \bigcup_{i \in \I} (\X^i)_t 
    \,\subseteq\, \bigcup_{i \in \I} \left(1 + \frac{t}{r} \right) \X^i 
    \,\subseteq\, \left(1 + \frac{t}{r} \right) \X.$$
    Taking volume on both sides, we conclude 
    $\Vol(\X_t) \le \Vol\left(\left(1 + \frac{t}{r} \right) \X\right) = \left(1 + \frac{t}{r} \right)^n \cdot \Vol(\X).$
    This shows that $\X$ satisfies the volume growth condition with $\alpha = 1$ and $\beta = 1/r$.
\end{proof}

\subsection{Volume growth condition under set union}
\label{Sec:GrowthUnion}

\begin{replemma}{Lem:VolGrowthUnion}
    Suppose $\X^i \subset \R^n$ is a compact body that satisfies the $(\alpha_i, \beta_i)$-volume growth condition for some $\alpha_i \in [1,\infty)$ and $\beta_i \in (0,\infty)$, for each $i \in \I$ in some finite index set $\I$.
    Let $q_\I$ be the probability distribution supported on $\I$ with density $q_\I(i) = \frac{\Vol(\X^i)}{\sum_{j \in \I} \Vol(\X^j)}$, for $i \in \I$.
    Then the union $\X = \bigcup_{i \in \I} \X^i$ satisfies the $(A,B)$-volume growth condition, where:
    \begin{align*}
        A &= \left(\max_{i \in \I} \alpha_i \right) \cdot \frac{\sum_{i \in \I} \Vol(\X^i)}{\Vol(\X)} \\
        B &= \E_{I \sim q_\I}\left[ \, \beta_I^n \, \right]^{1/n} = \left( \frac{\sum_{i \in \I} \Vol(\X^i) \cdot \beta_i^n}{\sum_{j \in \I} \Vol(\X^j)} \right)^{1/n} \le \max_{i \in \I} \beta_i.
    \end{align*}
\end{replemma}
\begin{proof}
As in the proof of Lemma~\ref{Lem:VolGrowthStarShaped}, we have $\X_t \subseteq \bigcup_{i \in \I} \X^i_t$ where $\X^i_t \equiv (\X^i)_t = \X^i \oplus B_t$.
Since each $\X^i$ satisfies the $(\alpha_i, \beta_i)$-volume growth condition, we have 
$\Vol(\X^i_t) \le \alpha_i \cdot (1 + t \beta_i)^n \cdot \Vol(\X^i)$ for all $i \in \I$.
Introducing the random variable $I \sim q_\I$ with density $q_\I(i) = \frac{\Vol(\X^i)}{\sum_{j \in \I} \Vol(\X^j)}$, we have:
\begin{align*}
    \frac{\Vol(\X_t)}{\Vol(\X)}
    &\le \frac{1}{\Vol(\X)} \cdot \sum_{i \in \I} \Vol(\X^i_t) \\
    &\le \frac{1}{\Vol(\X)} \cdot \sum_{i \in \I} \alpha_i \cdot (1 + t \beta_i)^n \cdot \Vol(\X^i) \\
    &= \frac{\sum_{j \in \I} \Vol(\X^j)}{\Vol(\X)} \cdot \E_{I \sim q_\I}\left[\alpha_I \cdot (1 + t \beta_I)^n \right] \\
    &\le \frac{\sum_{j \in \I} \Vol(\X^j)}{\Vol(\X)} \cdot \left(\max_{i \in \I} \alpha_i \right) \cdot \E_{I \sim q_\I}\left[(1 + t \beta_I)^n \right].
\end{align*}
For a random variable $U \in \R$ we denote its $L^p$-norm, $p \ge 1$, by
$\|U\|_{p} := \E[|U|^p]^{1/p}$.
Note that $B = \|\beta_I\|_n$ where $I \sim q_\I$.
Then we can write, using this notation and using triangle inequality:
\begin{align*}
    \E_{q_\I}[(1+t \beta_I)^n] 
    &= \|1+t \beta_I\|_n^n
    \le \left(\|1\|_n + t \|\beta_I\|_n \right)^n
    = (1 + t B)^n.
\end{align*}
Therefore, continuing the bound above, we obtain:
\begin{align*}
    \frac{\Vol(\X_t)}{\Vol(\X)}
    &\le \frac{\sum_{j \in \I} \Vol(\X^j)}{\Vol(\X)} \cdot \left(\max_{i \in \I} \alpha_i \right) \cdot \E_{I \sim q_\I}\left[(1 + t \beta_I)^n \right] \\
    &\le \frac{\sum_{j \in \I} \Vol(\X^j)}{\Vol(\X)} \cdot \left(\max_{i \in \I} \alpha_i \right) \cdot (1 + t B)^n.
\end{align*}
This shows $\X$ satisfies the $(A,B)$-volume growth condition where $A = \frac{\sum_{j \in \I} \Vol(\X^j)}{\Vol(\X)} \cdot \left(\max_{i \in \I} \alpha_i \right)$ and $B = \|\beta_I\|_n$.
In particular, $B \le \max_{i \in \I} \beta_i$.
\end{proof}

\subsection{Volume growth condition under set exclusion}
\label{Sec:VolGrowthExclusion}

\begin{replemma}{Lem:VolGrowthExclusion}
Let $\Y \subset \R^n$ be a compact body that satisfies the $(\alpha,\beta)$-volume growth condition for some $\alpha \in [1,\infty)$ and $\beta \in (0,\infty)$. 
Let $\X = \Y \setminus \Z$, where $\Z \subset \Y$ is an open set with $\Vol(\Z) < \Vol(\Y)$, and assume $\X$ is compact.
Then $\X$ satisfies the $(A, \beta)$-volume growth condition where $A = \alpha \cdot \frac{\Vol(\Y)}{\Vol(\X)}$.
\end{replemma}
\begin{proof}
Since $\X \subseteq \Y$, we have $\X_t \subseteq \Y_t$ for all $t \ge 0$ where $\X_t = \X \oplus B_t$ and $\Y_t = \Y \oplus B_t$.
Thus, since $\Y$ satisfies the $(\alpha,\beta)$-volume growth condition,
\begin{align*}
    \Vol(\X_t) \le \Vol(\Y_t) \le \Vol(\Y) \cdot \alpha \cdot (1 + t\beta)^n.
\end{align*}
Dividing by $\Vol(\X)$ yields
\begin{align*}
    \frac{\Vol(\X_t)}{\Vol(\X)} \le \frac{\Vol(\Y)}{\Vol(\X)} \cdot \alpha \cdot (1+t\beta)^n.
\end{align*}
This shows that $\X$ satisfies the $(A,\beta)$-volume growth condition, where $A = \alpha \cdot \frac{\Vol(\Y)}{\Vol(\X)}$.
\end{proof}

\paragraph{Acknowledgment.} 
The authors thank Yunbum Kook for helpful discussions.

\newpage
\addcontentsline{toc}{section}{References}
\bibliography{uniform.bbl}

\end{document}